\documentclass[final,3p,times,authoryear]{elsarticle}

\usepackage{amssymb}
\usepackage{amsmath}

\usepackage{xargs}                      
\usepackage[pdftex,dvipsnames]{xcolor}  
\usepackage[colorinlistoftodos,prependcaption,textsize=tiny]{todonotes}
\newcommandx{\unsure}[2][1=]{\todo[linecolor=red,backgroundcolor=red!25,bordercolor=red,#1]{#2}}
\newcommandx{\change}[2][1=]{\todo[linecolor=blue,backgroundcolor=blue!25,bordercolor=blue,#1]{#2}}
\newcommandx{\info}[2][1=]{\todo[linecolor=OliveGreen,backgroundcolor=OliveGreen!25,bordercolor=OliveGreen,#1]{#2}}
\newcommandx{\improvement}[2][1=]{\todo[linecolor=Plum,backgroundcolor=Plum!25,bordercolor=Plum,#1]{#2}}

\usepackage{amssymb}

\usepackage{amssymb,amsthm}
\usepackage{optidef}
\usepackage{multirow}
\usepackage[normalem]{ulem}
\useunder{\uline}{\ul}{}
\usepackage[]{hyperref}
\usepackage{placeins}

\newtheorem{lemma}{Lemma}
\newtheorem{prop}{Proposition}
\newtheorem{ex}{Example}

\newtheorem{theo}{Theorem}

\usepackage{booktabs}
\def\1{{1\kern-.3468em{ 1}}}

\usepackage{xcolor}

\usepackage{booktabs}

\usepackage{xspace}

\usepackage{subcaption}

\usepackage{amsfonts, amsmath, amssymb, mathtools, latexsym, bbm}

\usepackage{amsthm}

\newcommand{\hcunicamp}{HC - Unicamp}
\newcommand{\rqmodule}{Queue Management}
\newcommand{\overtimemodule}{Overtime and Cancellation}

\newcommand{\otschedmodule}{Operating Theatre Scheduling}

\newcommand{\rthreshold}{R}
\newcommand{\batchsizeq}{Q}
\newcommand{\lowq}{\underline{\batchsizeq}} 
\newcommand{\upq}{\overline{\batchsizeq}} 
\newcommand{\conflevelrq}{\alpha_{rq}}
\newcommand{\waittimereq}{\tau}

\newcommand{\stocprocperiod}{t}
\newcommand{\stocprocqueue}[1]{X_{#1}}
\newcommand{\stocprocbatches}[1]{Y_{#1}}
\newcommand{\stocprocremain}[1]{Z_{#1}}
\newcommand{\stocprocremainstate}{z}

\newcommand{\rvdemand}{D}
\newcommand{\drawnrvdemand}{d}
\newcommand{\roundbatches}{A}
\newcommand{\maxnumberbatch}{M}

\newcommand{\supportrq}{\Omega_{rq}}

\newcommand{\statespace}{S}

\newcommand{\steadystate}[1]{\pi_{#1}}
\newcommand{\transmatrix}[1]{P_{#1}}
\newcommand{\transmatrixval}[2]{p_{#1}^{#2}}

\newcommand{\completedsurgsidx}{k}
\newcommand{\kmax}{\ensuremath{\mathit{KMax_{\specialityidx}}\xspace}}

\newcommand{\blocklength}{T}
\newcommand{\safetythreshold}{t^{*}_{\specialityidx}}
\newcommand{\probstoppingthreshold}{\varepsilon_{\text{stop}}}
\newcommand{\plannedsurgs}{n}
\newcommand{\optplannedsurgs}{n^*}

\newcommand{\excesscost}{c}
\newcommand{\overtimepenalty}{p}
\newcommand{\cancelledsurgeries}{\ell}

\newcommand{\rvsurgdur}{D_{\specialityidx}}

\newcommand{\cumulativeduration}[1]{cd_{\specialityidx}^{(#1)}}
\newcommand{\numcompletedsurgs}{K_{\specialityidx}}

\newcommand{\probdistplan}{p_{\specialityidx}}
\newcommand{\cumulpdp}{F_{p_{\specialityidx}}}

\newcommand{\supportnv}{\Omega_{dur}}

\newcommand{\convolve}{\text{Convolve}}
\newcommand{\calccdf}{\text{CalculateCDF}}

\newcommand{\gain}[1]{G(#1)}
\newcommand{\cancellationcost}{C}

\newcommand{\marginalvalue}{g}
\newcommand{\monetarycancconst}{\kappa}
\newcommand{\cancellationfunction}{\ensuremath{\mathit{cost}}}

\newcommand{\survivalvector}{q_{\specialityidx}} 

\newcommand{\targetq}{{Q'}}

\newcommand{\blockdemand}{B}
\newcommand{\lowerblockdemand}{\underline{\blockdemand}_{\specialityidx}}
\newcommand{\upperblockdemand}{\overline{\blockdemand}_{\specialityidx}}

\newcommand{\blockset}{\mathcal{B}}
\newcommand{\blockmset}{\mathcal{B}^{\prime}}
\newcommand{\blockidx}{b}

\newcommand{\roomset}{OT}

\newcommand{\roomidx}{ot}

\newcommand{\specialityset}{\mathcal{S}}
\newcommand{\specialitymset}{\specialityset_{seq}}   
\newcommand{\specialityprset}{\specialityset_{prio}} 
\newcommand{\specialitytxset}{\specialityset_{tx}}
\newcommand{\specialityidx}{s}

\newcommand{\planninghorizonset}{\mathcal{D}}
\newcommand{\planninghorizonidx}{\mathit{d}}

\newcommand{\timeperiodset}{\mathcal{M}}

\newcommand{\blockvar}{x_{\roomidx,\blockidx,\specialityidx}}

\newcommand{\precvar}{y_{\roomidx,\planninghorizonidx,\specialityidx}}
\newcommand{\integervar}{w_{\blockidx,\specialityidx}}
\newcommand{\teamslackvar}{z_{\blockidx,\specialityidx}}
\newcommand{\anestslackvar}{a_{\blockidx}}

\newcommand{\conflevelot}{\alpha_{preproc}}

\newcommand{\numbersanests}{\xi}
\newcommand{\needanest}{\psi}
\newcommand{\specteams}{\phi}
\newcommand{\days}{|\planninghorizonset|}

\newcommand{\allowedrooms}{\lambda}

\newcommand{\anestconstant}{\gamma}

\newcommand{\modelname}{LT-OTS}

\usepackage{algorithm}
\usepackage{algpseudocode}

\usepackage{tikz}
\usetikzlibrary{arrows.meta, positioning, fit, calc, shapes.symbols}

\journal{European Journal of Operational Research}

\begin{document}

\begin{frontmatter}



\title{Integrated Framework for Long-term Elective Surgery Management under Uncertainty: From Strategic to Tactical Planning} 



\author[1]{João P. F. da Silva\corref{cor1}} \ead{joao.silva@ic.unicamp.br} 
\author[1]{Lucas Bortoletto} \ead{l173422@dac.unicamp.br}
\author[1]{Rafael C. S. Schouery} \ead{rafael@ic.unicamp.br}
\author[2,3]{Edilson F. Arruda} \ead{e.f.arruda@southampton.ac.uk}
\author[4]{Fabricio F. Santos} \ead{ffsantos@unicamp.br}
\author[4]{Ana Paula C. I. Alves} \ead{anapcia@unicamp.br}

\cortext[cor1]{Corresponding author. Tel.: (+44)(0)2380 597146.}

\affiliation[1]{organization={Institute of Computing, University of Campinas},
            city={Campinas},
            state={SP},
            country={Brazil}}

\affiliation[2]{organization={Centre for Healthcare Analytics, University of Southampton},
            city={Southampton},
            country={United Kingdom}}

\affiliation[3]{organization={University of Southampton Business School},
            city={Southampton},
            country={United Kingdom}}

\affiliation[4]{organization={Hospital de Clínicas, University of Campinas},
            city={Campinas},
            state={SP},
            country={Brazil}}

\begin{abstract}
In public healthcare systems, elective surgery waiting lists are a persistent management challenge, as hospitals must balance limited operating theatre capacity with uncertain and evolving demand.
Long-term planning for these lists requires more than assigning operating theatre (OT) time to medical specialities, as isolated allocation decisions may fail to stabilise waiting lists over time. 
Such planning must account for uncertainty in patient arrivals in the queue and surgery durations, which are often neglected in OT scheduling models. 
In this context, we propose a novel integrated framework for elective surgery management under uncertainty, linking strategic queue-control decisions with tactical OT scheduling. 
The framework combines three components: (i) a closed-loop control policy that determines the number of patients to schedule in each planning cycle; (ii) a model to account for overtime-related cancellation risks; and (iii) a mixed-integer linear programming formulation, informed by the previous components, to solve an OT Scheduling Problem. 
The resulting tactical schedules are therefore guided by the long-term evolution of speciality-specific queues and by cancellation-risk considerations. 
We evaluate the proposed framework through a case study at University Hospital Unicamp, a large Brazilian public referral hospital. 
The results show that integrating strategic and tactical decisions stabilises waiting lists, controls cancellation risks, and supports more transparent theatre-allocation decisions, while remaining computationally efficient for practical implementation.
\end{abstract} 




\end{frontmatter}

\newcommand{\citar}{\textcolor{red}{\textbf{[CITAR]}}}

\section{Introduction\label{sec:intro}}

Long-term planning for elective surgeries has attracted increasing attention from both public~\citep[e.g.,][]{Calegari2025} and private healthcare~\citep[e.g.,][]{Aissaoui2020} organisations, owing to its implications for access to care, service quality, and resource utilisation~\citep{AlAmin2024Review,Rachuba2024}. 
Initiatives such as the United Nations' Sustainable Development Goal 3~\citep{undp_good_health_undated} and Brazil's National Programme for Reducing Queues for Elective Surgeries~\citep{secom_pnrf_2024} illustrate the relevance of this issue by promoting medium (tactical) and long-term (strategic) actions aimed at improving access to surgical care. 

These government policies also underline the need for planning approaches that translate broad healthcare objectives into feasible decisions. 
Indeed, well-defined strategic goals provide guidance for shorter-term decisions, supporting the alignment between resource allocation, service targets, and broader healthcare priorities~\citep{Rathnayake2024, Bort2026}.

Waiting-list management is a central component of such strategic surgery planning, as it defines the point at which patients enter the surgical pathway~\citep{Daoud2026}. 
Specifically, the time patients spend on waiting lists directly affects access to care and the perceived quality of healthcare services~\citep{Freeman2018, Yu2026}. 
This issue is particularly relevant in public healthcare systems, such as Brazil's Unified Health System (\textit{Sistema Único de Saúde} -- SUS), where waiting-time requirements are linked to legal obligations and closely monitored by regulatory agencies~\citep{muricy2025filas}. 
Thus, waiting-list management is not only a matter of operational performance but also a critical aspect of healthcare equity~\citep{Melo2025,muricy2025filas}.

Although reducing waiting lists is often associated with maximising the number of patients treated~\citep[e.g.,][]{Barrera2018, Patro2022, Daoud2026}, focusing solely on throughput may lead to short-sighted planning decisions when queue dynamics are ignored. 
In real-world elective surgery systems, patient demand is uncertain, varies across surgical specialities, and evolves over time~\citep{Wang2014, Rahimi2020Review, Rachuba2024}. 
Ignoring such uncertainty may result in schedules that are misaligned with actual demand, leading to excessive waiting times or underutilisation of resources. 
Consequently, effective queue management requires more than allocating available operating theatre (OT) capacity to current demand; it must also account for stochastic patient arrivals over time and ensure that a hospital's service capacity remains aligned with the evolving state of the queues~\citep{Rathnayake2024}.

A second source of uncertainty arises from surgery duration~\citep{Wang2021, Calegari2025}. 
When defining how many surgeries should be planned within a block of OT time, assuming deterministic durations can produce plans that are difficult to implement in practice~\citep{Bansal2024}. 
Scheduling an excessive number of surgeries in a block may increase the risk of overtime, which in turn may lead to cancellations, schedule disruptions, and longer waiting times for patients who remain in the queue~\citep{AlAmin2024Review}. 
Conversely, overly conservative plans may leave scarce theatre time underused~\citep{Choudhary2024}. 
This balance has important implications for a strategic surgery planning, as it directly affects the number of planned surgeries that can be realistically performed, which in turn influences the long-term evolution of waiting lists.
Therefore, long-term elective planning should jointly consider uncertainty in patient demand and surgery duration, so that queue targets are achievable with a high level of confidence while the risk of overtime-related cancellations remains controlled.

Achieving control over the long-term evolution of waiting lists requires integrating strategic and tactical planning levels~\citep{Melo2025}. 
Specifically, strategic decisions define long-term targets for waiting-list management and determine how cancellation risks should be controlled, whereas tactical decisions translate these targets into OT allocations for surgical specialities. 
However, many studies on OT scheduling~\citep[e.g.,][]{Makboul2021Robust, MolinaPariente2026, Guo2026} focus primarily on the tactical allocation problem, often treating demand volumes or speciality requirements as deterministic inputs. 
Although such approaches can generate efficient schedules for a given planning period, they do not explicitly control the long-term evolution of speciality-specific waiting lists. 
This creates a gap between queue management objectives and the tactical decisions that determine how surgical resources are actually used. 
Accordingly, there is a need for planning approaches that connect strategic queue-control mechanisms with tactical OT allocation decisions.

One possible way to address this gap is to formulate a fully integrated optimisation model that simultaneously captures queue evolution, uncertainty in demand and surgery duration, cancellation risk, and OT allocation. 
By optimising all these aspects together, such a model could theoretically produce plans that are globally optimal with respect to the long-term defined objectives.
However, a monolithic approach may become computationally challenging and difficult to adapt to real-world hospital settings. 
This motivates the development of a modular system that integrates strategic and tactical decisions while preserving computational tractability and managerial interpretability.

In this paper, we propose a framework for long-term elective surgery planning under uncertainty. 
The framework consists of three interconnected modules:
(i) a waiting-list control module based on a modified $(R,Q)$ inventory policy~\citep{Axsater2015Inv} that determines the target number of patients to be scheduled in each planning cycle, thereby supporting long-term queue stability \citep{Shortle2018};
(ii) a cancellation-risk management module that uses a Newsvendor-based model~\citep{Arrow:1951,Petruzzi2011} to help define the number of surgeries to be planned within each OT block while controlling the risk of overtime-related cancellations; and
(iii) a tactical OT allocation module that uses a Mixed-Integer Linear Programming (MILP) model to assign surgical specialities to theatre blocks based on the targets generated by the previous modules. 

By linking these three components, the proposed framework is, to the best of our knowledge, the first to align long-term waiting-list management with medium-term resource allocation decisions, while explicitly accounting for key sources of uncertainty in elective surgery planning.

The proposed framework is evaluated through a case study at University Hospital Unicamp (\hcunicamp), one of Brazil's largest university hospitals and a major referral centre in the public healthcare system, serving 86 municipalities in the metropolitan region of Campinas \citep{ToledoTeixeira2023, Silva2025}.
The results from the experiments show how the framework can support waiting-list stabilisation, reduce the risk of overtime-related cancellations, and provide more transparent theatre-allocation decisions. 
Therefore, the paper contributes to the literature on elective surgery planning by offering an integrated yet modular approach that has practical relevance for real-world hospitals seeking to improve access to care while managing scarce resources effectively.

The remainder of this paper is organised as follows. 
Section~\ref{sec:lit_review} reviews the literature on the integration of strategic and tactical decisions in elective surgery planning, identifies the main research gaps, and positions the contributions of this paper. 
Section~\ref{sec:problem} describes the case study that motivates the proposed work. 
Section~\ref{sec:proposal} presents the proposed framework, detailing each of its components and their theoretical properties. 
Section~\ref{sec:experiments} reports the computational experiments conducted in the case-study setting. 
Finally, Section~\ref{sec:conclusion} concludes the paper by discussing the main implications of the proposed framework and outlining directions for future research.
%
\section{Related Works\label{sec:lit_review}}

Surgery scheduling problems have traditionally been classified into three planning levels with different sources of uncertainty: strategic, tactical, and operational~\citep{Rahimi2020Review, Cardoen2010, wang2021operating, Melo2025}. 
At the strategic level, long-term decisions define the desired split of the hospital's capacity for different surgical specialities and specify patient volumes to be admitted over the planning horizon~\citep{Burdett2023} --- a process commonly referred to as Case Mix Planning (CMP). 
Different objectives have been introduced at this level, such as maximizing profit~\citep{Fugener2015, Ma2013Mult} or minimizing long-term costs~\citep{Hof2015Casemix, Mcrae2020}.  The tactical level defines a cyclic schedule (often weekly) commonly known as Master Surgery Schedule (MSS)~\citep{Bovim2020,Siqueira2018INTO,Carneiro2025}. 
The MSS step allocates time blocks in OTs to surgical teams and is reviewed in a medium-term horizon (e.g., monthly or quarterly). 
Finally, the operational level focuses on the daily operation and derives a Surgical Case Schedule (SCS), which generally assigns individual patients to time blocks previously allocated to specific surgical teams at the tactical level~\citep{Vieira2025, Harris2022}.

Despite the large body of the literature in each of these decision levels, there are significant gaps to bridge~\citep{Rahimi2020Review, Melo2025}.
First, inter-level integration remains limited in scope, often focusing on short-term alignment between surgeries and post-surgical bed occupation, and generally relying on deterministic approaches~\citep{Melo2025,wang2021operating}. 
Second, strategies with a focus on guaranteeing waiting time requirements for patients still remain to be proposed~\citep{Siqueira2018INTO}. 
These gaps are possibly related to the complex propagation of uncertainty over different planning levels, which demands hybrid approaches that combine mathematical programming, stochastic modelling and queuing theory. 

The inherent interconnection between the strategic and tactical levels has motivated unified frameworks that link CMP and MSS decisions~\citep{Melo2025}.
These often adopt sequential or joint optimisation approaches in which case mix decisions act as input to surgery scheduling models~\citep{Ma2013Mult, Fugener2015}.
Modelling approaches typically consist of Mixed-Integer Linear Programming (MILP) formulations that focus on maximising hospital profit subject to OT and/or bed capacity constraints~\citep[e.g.,][]{Fugener2015}.
While these frameworks acknowledge uncertainty, this is often done through post-operative resource requirements, such as length of stay or ward occupancy~\citep{Moosavi2020, Schneider2020}, which are particularly relevant when downstream capacity is a binding concern.
In contexts where downstream units are sufficiently available, however, the integration of CMP and MSS decisions still requires attention to uncertainty arising within and before the surgical session itself.
This includes stochastic surgery durations, which directly affect OT utilisation, cancellations, idle time, and overtime risk~\citep{Choudhary2024}, as well as the explicit modelling of dynamic waiting lists across surgical specialities, which constitute the starting point of the surgical planning process~\citep{Daoud2026}.

The interconnection perspective can be broadened by incorporating additional uncertainty sources and dynamic elements.
For instance,~\citet{Liu2019} model elective admissions through a Markov Decision Process (MDP), highlighting the role of waiting-list size in guiding strategic decisions.
Despite its conceptual value, the framework disregards essential practical constraints, such as surgeon availability, block scheduling and case-mix heterogeneity, thereby limiting practical applicability.
Similarly,~\citet{Mcrae2020} demonstrate that explicitly modelling uncertainty and selecting appropriate aggregation levels can significantly improve CMP outcomes and enable joint CMP-MSS optimisation.
Nevertheless, waiting-list dynamics remain outside of the modelling scope, implying no full integration of strategic and tactical decisions. 
Furthermore, their reliance on Sample Average Approximation (SAA)~\citep{Hu2010} may lead to representations that may understate distributional variability. 
These studies underscore the value of integrating CMP and MSS, yet they also reveal persistent shortcomings: uncertainty is represented only selectively across planning levels, waiting-list requirements are rarely modelled explicitly, and key operational constraints required for real-world implementation remain overlooked.

Uncertainty in surgery duration is another important issue that, when disregarded, may result in unpredictable idle or overtime in OTs~\citep{shehadeh2022stochastic}. 
Whilst overtime is costly and can lead to  cancellations that compromise the quality of care, idle time reflects an inadequate use of OT capacity~\citep{shehadeh2022stochastic, Harris2022, Kayvanfar2025}. 
With the aim of maximising hospital revenue while preventing excessive overtime,~\citet{Barz2015} propose a tactical patient admission and scheduling problem that derives a CMP under multiple resource constraints and stochastic clinical progression. 
The approach utilises a newsvendor heuristic to reserve capacity for emergency arrivals, and an MDP model to select the patient cohort to be admitted whilst considering probabilistic transitions between clinical states. 
The model, however, excludes the possibility of cancellations or postponements driven by overtime, as it assumes that the hospital always disposes of the resources to treat each admitted patient. 
Moreover, by treating resource consumption as deterministic, the framework effectively ignores uncertainty in surgery duration --- one of the main drivers of operational instability. 
As a result, the model underestimates the practical interplay between stochastic surgery durations, cancellations and overtime, limiting its applicability to integrated CMP-MSS decision-making.

Furthermore,~\citet{Zhang2020} study a tactical surgery planning with the explicit aim of mitigating overtime risk when allocating elective patients to OTs over an upcoming planning horizon. It is worth noting, however, that long-term waiting time management is not  considered. Their model accounts for stochastic surgery durations and adopts a risk-averse objective that jointly captures both the probability and the expected magnitude of overtime, extending earlier approaches that consider only one of these dimensions~\citep{Shylo2013, Rath2017}.
Although they derive the durations from real-world data, eventually they approximate it through a set of discrete scenarios, limiting the full distributional representation.
A similar treatment appears in~\citet{Shehadeh2026}, which introduces a Distributionally Robust Optimisation (DRO) approach to address uncertainty in surgery durations.
While their approach minimises worst-case expected cost over an ambiguity set defined by known moments and support, it likewise relies on scenario-based representations.
Accurately capturing surgery-duration variability requires a large number of scenarios, substantially increasing computational burden and potentially undermining the practical advantages of DRO and related methods.

These limitations regarding the integration of overtime management and its associated cancellation risks into surgical planning frameworks highlight important research gaps.
There remains a need for models that explicitly incorporate cancellation dynamics induced by overtime and that provide a more thorough representation of surgery-duration uncertainty beyond scenario-based approximations. 
A key challenge addressed in this paper is the development of a modular approach that captures full distributional characteristics to build a block planning model. 
The output of this model is then used as input to an integrated CMP-MSS approach, thereby linking the two decision layers sequentially rather than embedding them in a single monolithic formulation.

An additional and often underexplored dimension in integrated surgical planning concerns the incorporation of practical constraints inherent to real-world healthcare settings. 
Certain surgical activities, such as organ transplantation, impose rigid scheduling requirements that are difficult to reconcile with standard optimisation assumptions~\citep{Silva2025}. 
Furthermore, public hospitals in developing countries (e.g., Brazil) operate under severe resource constraints, high and volatile demand, and heterogeneous case mixes, all of which exacerbate the complexity of surgical planning~\citep{Carneiro2025, Siqueira2018INTO, Silva2025}. 

These contextual factors challenge the scalability and transferability of many optimisation-based approaches, which often presume stable resources, predictable demand and homogeneous patient populations. 
Indeed, as noted by~\citet{Melo2025}, much of the existing literature is calibrated to relatively small hospitals in developed regions, limiting its relevance for large, resource-constrained systems. 
Ignoring such operational realities risks producing models that are theoretically stable but difficult to implement in practice, underscoring the need for scheduling frameworks that explicitly accommodate institutional, organisational and contextual constraints.

This paper addresses the identified gaps in the elective surgical planning literature by proposing a novel integrated framework that links strategic waiting-list control and tactical OT allocation while explicitly accounting for dynamic surgical demand and stochastic surgery durations. 


By jointly addressing relevant sources of uncertainty within a unified framework, this work introduces an innovative and integrated surgical planning approach that advances the literature by bridging multiple relevant gaps. 
The approach also bridges the gap between theory and practice by bringing forward a practically grounded and empirically validated framework for long-term management of surgery queues within complex healthcare settings.
%
\section{Description of the Case Study\label{sec:problem}}

\hcunicamp~is a university hospital located in Campinas, Brazil~\citep{hc_unicamp_historia}. 
It is a high-complexity public hospital, serving approximately $6.5$ million inhabitants from Campinas' metropolitan region, that handles around $500{,}000$ patient cases per year. 
Of these, roughly $10{,}000$ correspond to elective surgical procedures, with the remainder involving outpatient consultations and emergency care~\citep{hc_unicamp_dados}. 
In addition to providing healthcare services, this hospital operates as a teaching and research hospital, training healthcare professionals and conducting medical and interdisciplinary research.

Elective surgical patients follow a complex multi-stage care pathway (see Figure~\ref{fig:patient_flow}), encompassing pre-, intra- and post-operative phases.
This process begins with waiting-list management, initiated at local medical consultations and coordinated through the regional bed regulation centre~\citep{central_reguladora}, a government agency responsible for regulating hospital access across Campinas Metropolitan Region.
Patient entry into \hcunicamp~typically occurs via scheduled outpatient appointments, although some emergency cases subsequently become elective and are incorporated into the same planning process.
Once referred, patients are assigned to waiting lists according to their medical subspeciality.
Decisions at this stage directly determine the CMP (i.e., strategic layer), influencing waiting times, speciality-level demand balance and the long-term utilisation of OT capacity~\citep{Siqueira2018INTO}.
%
\begin{figure}[htb]
    \centering
    \includegraphics[width=0.99\textwidth]{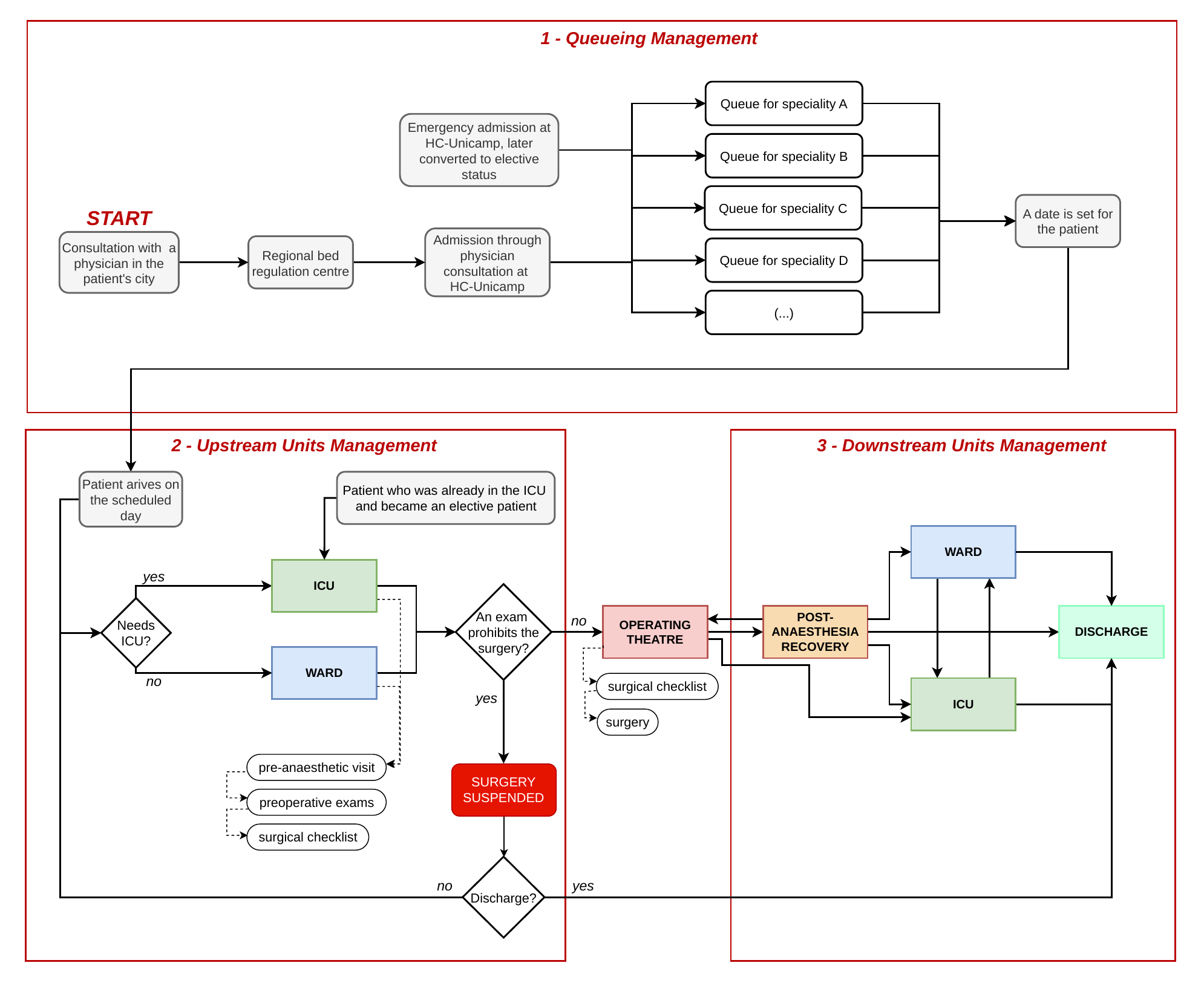}
    \caption{Patient flow at \hcunicamp.\label{fig:patient_flow}}
\end{figure}

On the day of surgery, patients enter the upstream unit management phase, where clinical reassessment is performed.
Based on clinical requirements, Intensive Care Unit (ICU) support may be required pre- or postoperatively, necessitating prior bed reservation; otherwise, postoperative ward bed availability must be ensured.
Patients then undergo pre-anaesthetic evaluation, diagnostic tests and surgical safety checks.
If contraindications arise, the procedure is suspended and the patient is either discharged or rescheduled, according to clinical severity.
If cleared, the patient proceeds to the OT, where the surgical procedure and subsequent theatre cleaning are performed.
In the present study, ward and ICU capacity are sufficient and therefore do not constitute binding constraints, allowing the analysis to focus on waiting-list management, uncertainty in surgery duration and OT allocation.

Once in the OT, the patient is managed by a dedicated surgical team and anaesthetist. 
This stage includes final safety checks, the surgical procedure itself and post-operative cleaning. 
Here, uncertainty in surgery duration plays a central role, as it directly affects idle time, overtime, and the risk of cancellations. 
Indeed, it is well known that surgery duration variability has a significant impact on OT utilisation and overall system throughput \citep[e.g.,][]{Pandit2011}, therefore meriting explicit consideration in the proposed framework.

Following surgery, patients are transferred to the Post-Anaesthesia Recovery Unit (PARU), with direct ICU transfer occurring only in exceptional cases. 
After stabilisation, patients are admitted to a ward or ICU bed and subsequently discharged. 

\hcunicamp~has $15$ OTs dedicated to elective surgeries: $13$ located in the Elective Surgery Centre (ESC) and $2$ in the Ambulatory Surgery Centre (ASC), which operates exclusively during morning shifts. These theatres are heterogeneous with respect to infrastructure and equipment, leading to compatibility restrictions across medical subspecialities.

The planning horizon comprises approximately $20$ working days per month, each divided into two $6$-hour shifts (morning and afternoon). 
Henceforth, we refer to a theatre shift as a \textit{block}.
Hospital policy assigns each block to a single subspeciality, while allowing the same subspeciality to occupy multiple blocks within the same theatre, as long as the compatibility conditions are satisfied.
Sequential allocation is therefore permitted, enabling a subspeciality assigned to a morning block to also occupy the corresponding afternoon block. Sequential blocks are preferred by hospital managers as they improve logistical operations.

Each month, OTs must be allocated across $40$ medical subspecialities.
Relevant subsets include: (i) those that mandatorily require two consecutive blocks in the same theatre due to the extended duration of their procedures; (ii) priority subspecialties defined by hospital management, which take precedence in allocation decisions in accordance with institutional rules; and (iii) transplant subspecialities, which impose additional operational constraints.

Block allocation feasibility depends mainly on the availability of surgical teams and anaesthetists.
A subspeciality can only be assigned to a block if an appropriate team is available on the corresponding day and shift.
Moreover, the number of simultaneous (i.e., in the same day-shift) allocations requiring anaesthesia cannot exceed the number of anaesthetists available in each period.

\begin{figure}[htb]
    \centering
    \includegraphics[width=0.9\textwidth]{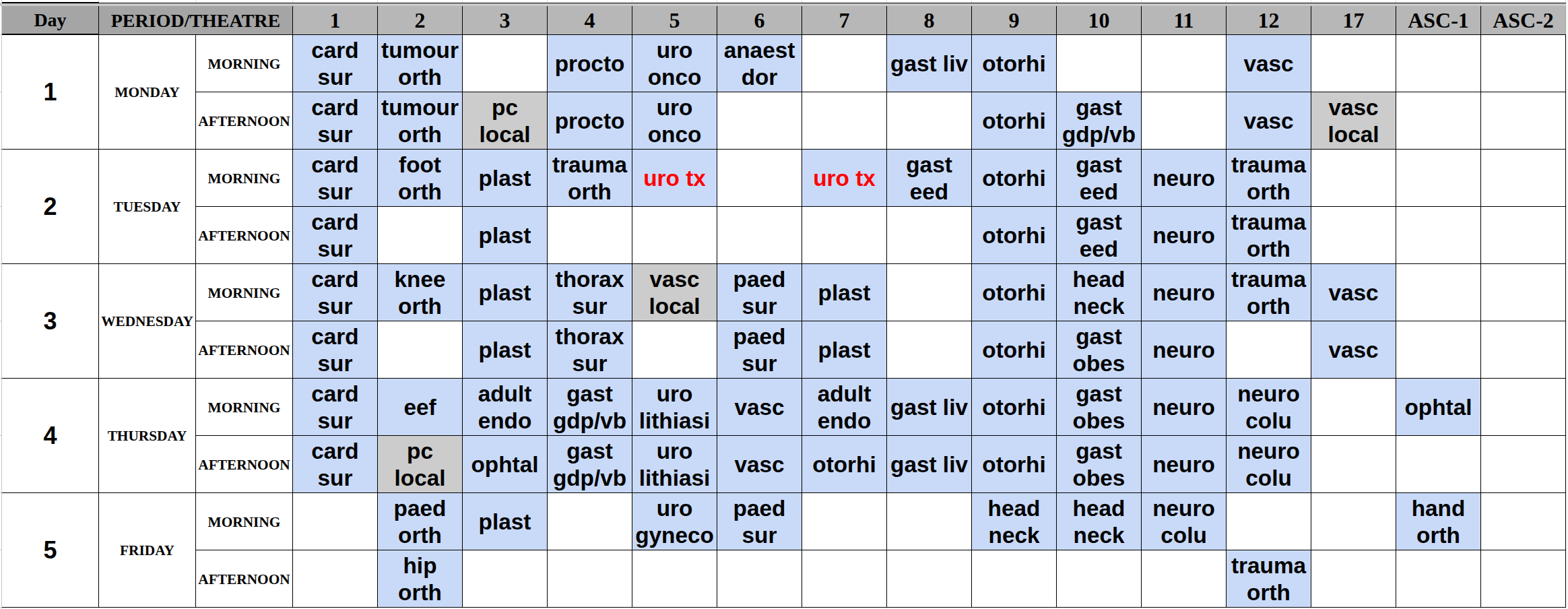}
    \caption{Excerpt from a week's OT schedule.\label{fig:ex_escala}}
\end{figure}

Figure~\ref{fig:ex_escala} presents an excerpt from a monthly OT schedule. Rows correspond to day and shift and columns to theatres; allocated blocks are colour-coded (blue for procedures requiring anaesthesia and grey otherwise). 
Unallocated (blank) slots reflect the complexity of the constraint set. 
Allocations shown in red correspond to transplant surgeries, which require the simultaneous allocation of two theatres within the same day and shift.

Historically, schedules such as those shown in Figure~\ref{fig:ex_escala} have been manually constructed using only electronic spreadsheets, often requiring more than two full working days. This manual process is time-consuming and does not guarantee optimality with respect to institutional policies. Automating the scheduling process therefore reduces administrative burden and enables the generation of optimal solutions under complex constraints.

In addition to the impact of automating surgery scheduling in a reference hospital in Brazil, this paper provides an approach that is firmly grounded in reality and is therefore adaptable to other general hospitals. By jointly managing waiting-list dynamics, surgery duration uncertainty and OT allocation, our proposed approach provides a coherent solution capable of satisfying waiting-time requirements while mitigating idle time, overtime and cancellation risks.

%
\section{Proposed Framework\label{sec:proposal}}

Rather than formulating a single monolithic model, we propose a framework that decomposes the long-term elective surgery problem into three inter-related modules. 
The first two modules address strategic decisions related to queue and overtime management, and their outputs are then used as inputs to the tactical OT scheduling model. 
Figure~\ref{fig:framework} summarises the proposed framework.

\begin{figure}[htb]
    \centering
    \includegraphics[width=0.99\textwidth]{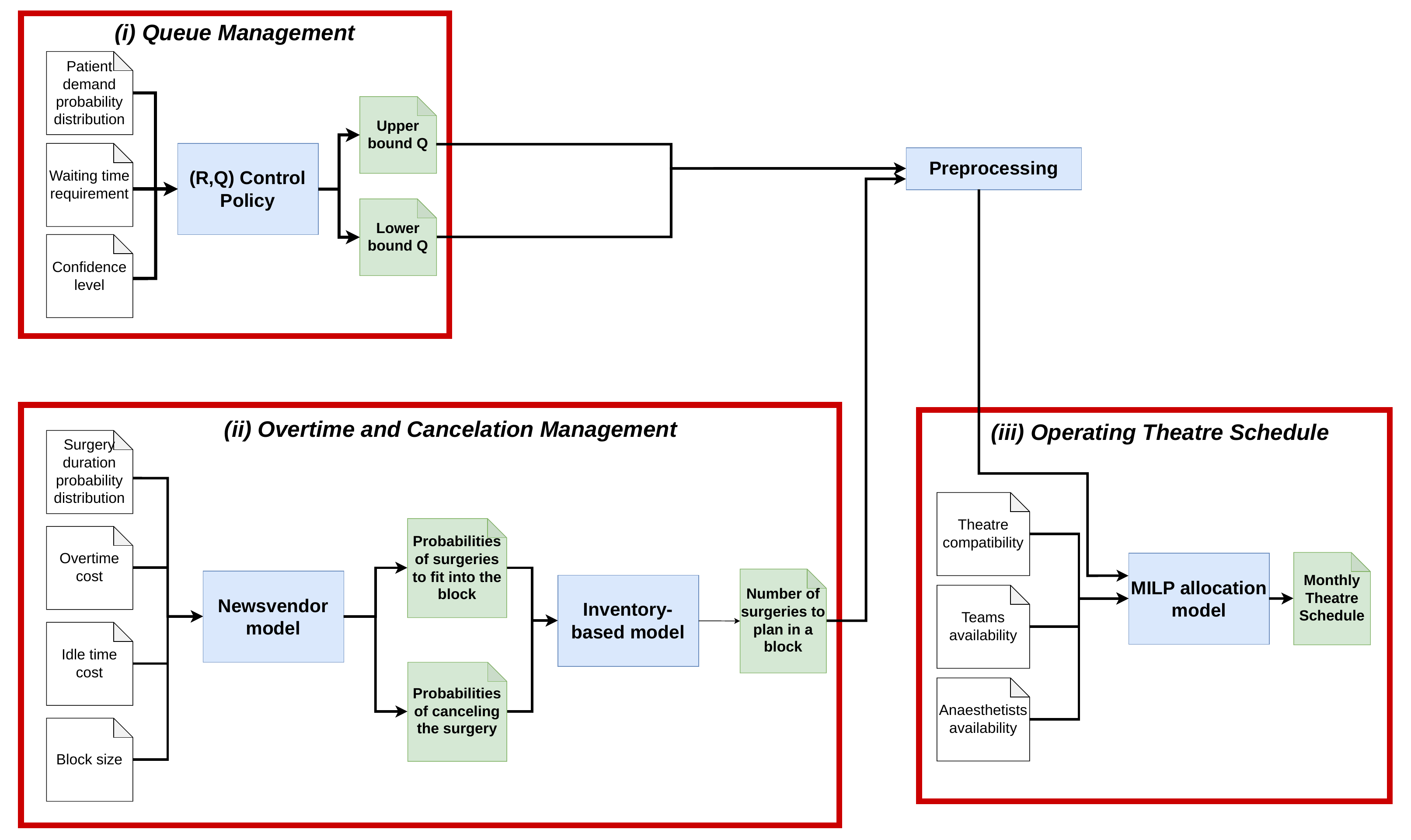}
    \caption{Proposed framework.\label{fig:framework}}
\end{figure}

The \rqmodule~module establishes targets for the number of surgeries to be performed at each planning period (e.g., monthly). 
Specifically, a modified $(R,Q)$ policy defines a lower and an upper bound on the number of patients to be admitted for surgery during the planning period for each medical subspeciality, while ensuring that system-wide waiting-time requirements are satisfied in spite of the uncertainty in patient arrivals.

The second module balances idle time, overtime and its consequential risk of surgery cancellation. 
Considering the uncertainty in surgery duration, a Newsvendor-based model guides dynamic decisions on whether to perform or postpone the next surgery, given the time remaining in the session. This model produces a probability distribution of the number of surgeries carried out for any number of planned procedures. These distributions are then evaluated  in an inventory-based formulation that determines the optimal number of surgeries to book in a block.


The outputs of the first two modules are then combined to provide strategic inputs to the third module, which solves an OT Scheduling Problem (OTSP), i.e., a tactical scheduling model that allocates OT blocks to surgical subspecialities.
Specifically, given the probability distribution of the optimal number of surgeries per session (module 2), and the bounds on the number of surgeries to perform in each period (module 1), one can derive lower and upper limits on the number of blocks to be allocated to each surgical subspeciality. 
Hence, theatre allocation decisions are informed not only by resource availability and OT compatibility, but also include constraints to guarantee that waiting-time requirements can be met despite the uncertainties in the demand and duration of surgeries. 
Additionally, the scheduling model remains deterministic, as the uncertainty is captured in the parameters derived from the first two modules.

Overall, the proposed framework provides a consistent approach to elective surgery planning across multiple decision levels.
By explicitly incorporating uncertainty into the strategic planning stages, it supports more realistic scheduling decisions that are guaranteed to align to waiting-time requirements.
The following subsections present each module in detail. 

%
\subsection{Module 1: Long-Term Waiting List Management\label{sec:waitinglist}}

As discussed in Section~\ref{sec:lit_review}, CMP constitutes the strategic decision layer responsible for determining the volumes of elective patients to be admitted over a given planning horizon.
In contrast to much of the existing literature, which primarily selects these volumes to optimise financial or capacity-related objectives, this work adopts a stability-oriented perspective.
We aim to enable decision-makers to impose maximum acceptable waiting times (e.g., six months) and to guarantee, in the long run, that such targets are met without allowing waiting lists to grow unboundedly. 
This is especially important in public healthcare systems, such as Brazil's SUS, where there is a legal duty to provide care to all patients who need it~\citep{muricy2025filas}.
Therefore, ensuring equitable access to care and adherence to waiting time standards is paramount.

In line with this objective, this section introduces the theoretical foundations of the modified $(R,Q)$ control policy proposed to manage elective surgery waiting lists under stochastic demand.
We first characterise the dynamics of the underlying stochastic process and analyse its long-run behaviour. 
Subsequently, we derive analytical upper and lower bounds on the batch size $Q$ that ensures compliance with prescribed waiting-time requirements, thereby establishing the strategic input for the integrated planning framework.
Table~\ref{tab:symb_rq} summarises the notation used in this section.

\begin{table}[htb]
    \centering
    \caption{Symbols used in the long-term waiting list management model.\label{tab:symb_rq}}
    \scriptsize
    \begin{tabular}{cl}
    \toprule
    \textbf{Index} & \textbf{Description} \\
    \midrule
    $\stocprocperiod$ & Planning period index \\
    $\stocprocremainstate$ & State index of the remainder process \\
    $i,j$ & Generic state indices \\
    \midrule

    \textbf{Set} & \textbf{Description} \\
    \midrule
    $\statespace$ & State space of process $\stocprocqueue{\stocprocperiod}$ \\
    $\statespace_{\stocprocbatches{}}$ & State space of process $\stocprocbatches{\stocprocperiod}$ \\
    $\statespace_{\stocprocremain{}}$ & State space of process $\stocprocremain{\stocprocperiod}$ \\
    $\supportrq$ & Support of the demand random variable \\
    \midrule

    \textbf{Parameter} & \textbf{Description} \\
    \midrule
    $\rthreshold$ & Reorder threshold of the modified $(R,Q)$ policy \\
    $\batchsizeq$ & Batch size under the modified $(R,Q)$ policy \\
    $\waittimereq$ & Waiting-time target \\
    $\conflevelrq$ & Probability threshold for the waiting-time requirement \\
    $\upq$ & Upper bound on feasible batch sizes \\
    $\lowq$ & Lower bound on feasible batch sizes \\
    $\maxnumberbatch$ & Maximum number of batch arrivals in one period \\
    \midrule

    \textbf{Random Variable / Stochastic Process} & \textbf{Description} \\
    \midrule
    $\stocprocqueue{\stocprocperiod}$ & Queue-length process \\
    $\rvdemand$ & Demand between two consecutive planning periods \\
    $\drawnrvdemand$ & Realisation of $\rvdemand$ \\
    $\stocprocbatches{\stocprocperiod}$ & Batch-level queue process \\
    $\stocprocremain{\stocprocperiod}$ & Remainder process associated with $\stocprocqueue{\stocprocperiod}$ \\
    $\roundbatches$ & Number of batch arrivals in one period \\
    $\omega$ & Waiting time of an arriving patient \\
    \midrule

    \textbf{Function} & \textbf{Description} \\
    \midrule
    $f_{\rvdemand}(\cdot)$ & Probability mass function of $\rvdemand$ \\
    $F_{\rvdemand}(\cdot)$ & Cumulative distribution function of $\rvdemand$ \\
    $\roundbatches(\stocprocremainstate)$ & Number of batch arrivals given remainder state $\stocprocremainstate$ \\
    $\steadystate{\stocprocremain{}}(\cdot)$ & Steady-state distribution of $\stocprocremain{\stocprocperiod}$ \\
    $\steadystate{\stocprocbatches{}}(\cdot)$ & Steady-state distribution of $\stocprocbatches{\stocprocperiod}$ \\
    $\transmatrix{\stocprocremain{}}$ & Transition matrix of process $\stocprocremain{\stocprocperiod}$ \\
    $\transmatrix{}^{\stocprocbatches{}}$ & Transition matrix of process $\stocprocbatches{\stocprocperiod}$ \\
    $\transmatrixval{\stocprocremainstate\stocprocremainstate'}{\stocprocremain{}}(k)$ & Conditional transition probability of $\stocprocremain{\stocprocperiod}$ given demand $k$ \\
    $\transmatrixval{ij}{}$ & Generic transition probability entry \\
    \bottomrule
    \end{tabular}
\end{table}

\subsubsection{$(R,Q)$ surgery scheduling policy: long-term distributions\label{sec:rqtheory}}

Let $\stocprocqueue{\stocprocperiod}, \, \stocprocperiod \ge 0$, be a stochastic process representing an elective surgery queue with stationary and independent increments over discrete planning periods, with $\stocprocqueue{\stocprocperiod}$ representing the number of patients in the system at period $\stocprocperiod$. 
Assume that the queue is controlled via a modified $(R,Q)$ policy \citep[e.g.,][]{Axsater2015Inv} whereby a fixed batch of $\batchsizeq \in \mathbb N$ surgeries is processed at time period $\stocprocperiod$ whenever $\stocprocqueue{\stocprocperiod} \ge \rthreshold$  at the start of the period, where $\mathbb N$ is the set of positive integers. 
In contrast to the classical $(R,Q)$ policy, which processes as many batches as necessary to bring the inventory below the reorder point $R$, our policy can only process a batch at a time and relies on long-term properties to ensure that the process $\stocprocqueue{\stocprocperiod}, \, \stocprocperiod \ge 0$ is stable and guarantees prescribed waiting time requirements. To facilitate the discussion, we will assume that $R=Q$ for the remainder of the session.

Let $\rvdemand$ be a non-negative integer random variable representing the overall demand between two consecutive planning periods, with probability distribution $f_{\rvdemand}: \Omega \to [0,1), \, \Omega \subset \mathbb Z_+$, and cumulative distribution $F_{\rvdemand}: \mathbb Z_+ \to [0,1]$ given by
\[
F_{\rvdemand}(d) = \sum_{l=0}^{\drawnrvdemand} f_{\rvdemand}(l),
\]
where $\mathbb Z_+$ denotes the set of non-negative integers. This implies that the dynamics of process $\stocprocqueue{\stocprocperiod}, \, \stocprocperiod \ge 0$ can be expressed as:
\begin{equation} \label{eq:dynamics}
\stocprocqueue{\stocprocperiod+1} = \stocprocqueue{\stocprocperiod} + \drawnrvdemand - \batchsizeq \1_{\{\stocprocqueue{\stocprocperiod} \ge \rthreshold\}},
\end{equation}
where $\1_{\{A\}}$ is the indicator function, which is equal to $1$ if $A$ holds and is $0$ otherwise, and $\drawnrvdemand$ is a value drawn from $f_{\rvdemand}$ that represents the surgical demand in period $\stocprocperiod$. The state space of process $\stocprocqueue{\stocprocperiod}, \, \stocprocperiod \ge 0$ is given by $\statespace = \{ 0, \, 1 , \, 2, \, \ldots \}$, where $\stocprocqueue{\stocprocperiod}=i$ indicates that there are $i$ patients waiting for elective surgery.

Letting 
\begin{equation} \label{eq:defYZ}
\stocprocbatches{\stocprocperiod} = \left \lfloor \dfrac{\stocprocqueue{t}}{Q} \right \rfloor, \qquad Z_t = \stocprocqueue{t} \, \% \,  \batchsizeq,
\end{equation}
where $\%$ represents the modulus operator, we can write:
\[
\stocprocqueue{\stocprocperiod} = Q \stocprocbatches{\stocprocperiod} + \stocprocremain{\stocprocperiod},
\]
where $\stocprocbatches{\stocprocperiod}, \, \stocprocperiod \ge 0$ is the integer number of batches required to bring the $X_t$ instantaneously below $\rthreshold$ and $\stocprocremain{\stocprocperiod} \in \{0, 1, \ldots, \batchsizeq-1\}$ is the number of patients that would remain after removing $\stocprocbatches{\stocprocperiod}$ batches of $\batchsizeq$ patients instantaneously.

Note that, from the definitions above, the state space of process $\stocprocbatches{\stocprocperiod}, \, \stocprocperiod \ge 0$ is $\statespace_{\stocprocbatches{}} = \mathbb Z_+$, while $\statespace_{\stocprocremain{}} = \{0, \, 1, \, \ldots, \, Q-1\}$ is the state space of process $\stocprocremain{\stocprocperiod}, \, \stocprocperiod \ge 0$. In the next result, we will derive the long-term distribution of process $\stocprocremain{\stocprocperiod}, \, \stocprocperiod \ge 0$.

\begin{prop} \label{prop:steadystate}
Let $\steadystate{\stocprocremain{}}: \statespace_{\stocprocremain{}} \to [0,1]$ be the steady state distribution of process $\stocprocremain{\stocprocperiod}, \, \stocprocperiod \ge 0$. Then, it holds that 
\begin{equation} \label{eq:steadyZ}
\steadystate{\stocprocremain{}}(\stocprocremainstate) = \dfrac{1}{\batchsizeq}, \forall \stocprocremainstate \in \statespace_{\stocprocremain{}}.
\end{equation}
\end{prop}
\begin{proof}
The proof closely follows that of \citet[Proposition 5.1]{Axsater2015Inv}. Let $\transmatrix{\stocprocremain{}}$ be the transition matrix describing the dynamics of process $\stocprocremain{\stocprocperiod}, \, \stocprocperiod \ge 0$. Assume that $\stocprocremain{\stocprocperiod} = \stocprocremainstate, \, \stocprocremainstate \in \statespace_{\stocprocremain{}}$ at the start of period $\stocprocperiod \ge 0$, and let $d=k$ be the number of new patient arrivals in period $\stocprocperiod$. Hence, it follows that 
\[
\stocprocremain{\stocprocperiod+1} =  (\stocprocremainstate + k) \,\%\, \batchsizeq  = (\stocprocremainstate + k + \batchsizeq \stocprocbatches{\stocprocperiod}) \,\%\, \batchsizeq  = (\stocprocqueue{\stocprocperiod} + k) \,\%\, \batchsizeq,
\] 
where the last equality come from the definitions in Eq.~\eqref{eq:defYZ}. Therefore, given that the realisation of the random demand is $\rvdemand=k$, it is evident that $\transmatrixval{\stocprocremainstate\stocprocremainstate'}{\stocprocremain{}}(k) := P( \stocprocremain{\stocprocperiod+1} = \stocprocremainstate' | \stocprocremain{\stocprocperiod} = \stocprocremainstate, \rvdemand=k )$ equals one for exactly one value of $\stocprocremainstate$ and zero otherwise. Consequently, we can write:
\[
\sum_{\stocprocremainstate=0}^{\batchsizeq-1} \transmatrixval{\stocprocremainstate\stocprocremainstate'}{\stocprocremain{}} = \sum_{\stocprocremainstate=0}^{\batchsizeq-1} \mathbb{E}_{k} \left\{ \transmatrixval{\stocprocremainstate\stocprocremainstate'}{\stocprocremain{}}(k) \right\} = \mathbb{E}_k \left\{  \sum_{\stocprocremainstate=0}^{\batchsizeq-1} \transmatrixval{\stocprocremainstate\stocprocremainstate'}{\stocprocremain{}}(k) \right\} = \mathbb{E}_k \{ 1 \} = 1.
\]
The expression above  implies that the Markov chain describing the dynamics of process $Z_t, \, t \ge 0$ is doubly stochastic, and therefore has a uniform steady state distribution. Indeed, substituting \eqref{eq:steadyZ} in the steady state equations:
\[
\pi_Z(z) = \sum_{\stocprocremainstate'=0}^{\batchsizeq-1} \pi_Z(z')  \transmatrixval{\stocprocremainstate'\stocprocremainstate}{\stocprocremain{}},
\]
one can easily see that the steady state distribution satisfies Eq. \eqref{eq:steadyZ}.
\end{proof}

The result in Proposition \ref{prop:steadystate} is very important, as it shows that the steady state distribution of process $\stocprocremain{\stocprocperiod}, \, \stocprocperiod \ge 0$, is uniform regardless of the distribution of patient arrivals. Next, we will use this fact to derive the steady state distribution of process $\stocprocbatches{\stocprocperiod}, \, \stocprocperiod \ge 0$. 

\subsubsection{The dynamics of process $\stocprocbatches{\stocprocperiod}, \, \stocprocperiod \ge 0$\label{sec:batchprocess}}

Assume that process $\stocprocremain{\stocprocperiod}, \, \stocprocperiod \ge 0$ has reached equilibrium, which is also equivalent to supposing that $P(\stocprocremain{0} = \stocprocremainstate) = \steadystate{\stocprocremain{}}(\stocprocremainstate) = \dfrac{1}{\batchsizeq}$ \citep[e.g.,][]{Bremaud2020}. Recalling that one batch is processed whenever $X_t \ge R$, we can describe the dynamics of process $\stocprocbatches{\stocprocperiod}, \, \stocprocperiod \ge 0$ as:
\begin{equation} \label{eq:dynY}
\stocprocbatches{\stocprocperiod+1} = \begin{cases}
\stocprocbatches{\stocprocperiod} - 1 + \roundbatches, & \text{if} \; \stocprocbatches{\stocprocperiod} > 0, \\
\roundbatches; \; & \text{if} \; \stocprocbatches{\stocprocperiod} = 0,
\end{cases}
\end{equation}
where $\roundbatches$ is the rounded down number of batch arrivals at time $\stocprocperiod$. To fully characterize the random variable $\roundbatches$, it is necessary to define
\begin{equation} \label{eq:Az}
\roundbatches(\stocprocremainstate) = \left \lfloor \dfrac{\rvdemand + \stocprocremainstate}{\batchsizeq} \right \rfloor,
\end{equation}
which represents the rounded down number of batch arrivals at time $\stocprocperiod$, given that $\stocprocremain{\stocprocperiod} = \stocprocremainstate$. Then, in view of Proposition~\ref{prop:steadystate}, it follows that:
\begin{equation} \label{eq:batcharr}
P(\roundbatches = k ) = \dfrac{1}{\batchsizeq} \sum_{\stocprocremainstate=0}^{\batchsizeq-1} P(\roundbatches(\stocprocremainstate) = k ).    
\end{equation}

Therefore, from \eqref{eq:dynY} we have:
\begin{equation} \label{eq:ydynamicY}
\transmatrixval{ij}{\stocprocbatches{}} = P(\stocprocbatches{\stocprocperiod+1} = j | \stocprocbatches{\stocprocperiod} = i ) = \begin{cases}
    P\left( \roundbatches = j  \right), & \text{if} \; \stocprocbatches{\stocprocperiod} = 0, \\
    P\left( \roundbatches = j-i + 1 \right),  & \text{if} \; \stocprocbatches{\stocprocperiod} \ge 1.
\end{cases}
\end{equation}

Note that Eq. \eqref{eq:dynY}-\eqref{eq:ydynamicY} correspond to the dynamics in \citet[Section 6.1.2]{Shortle2018}. Similarly to \citet{Shortle2018}, the transition matrix of process $\stocprocbatches{\stocprocperiod}, \, \stocprocperiod \ge 0$ can be written as:
\begin{equation} \label{eq:transexplicit}
\transmatrix{}^{\stocprocbatches{}} = [\transmatrixval{ij}{}] = \left[
\begin{array}{cccccccc}
k_0 & k_1 & k_2 & \ldots & k_M & 0 & 0 & \ldots \\
k_0 & k_1 & k_2 & \ldots & k_M & 0 & 0 & \ldots \\
0 & k_0 & k_1 & k_2 & \ldots & k_M & 0 & \ldots \\
0 & 0 & k_0 & k_1 & k_2 & \ldots & k_M & \ldots \\
\vdots & \vdots & \vdots & \vdots & \vdots & \vdots & \vdots & 
\end{array}
\right],
\end{equation}
where $\maxnumberbatch$ is the maximum number of new batch arrivals, i.e., the maximum value that random variable $\roundbatches$ can assume.

The steady state probabilities of process $\stocprocbatches{\stocprocperiod}, \, \stocprocperiod \ge 0$ can then be recursively evaluated by solving:
\begin{equation} \label{eq:ssY}
\steadystate{\stocprocbatches{}}^T = \steadystate{\stocprocbatches{}}^T P,
\end{equation}
with the contour condition that $\steadystate{\stocprocbatches{}}(0) = 1 - \rho $ \citep[Eq. (6.12)]{Shortle2018}, 
%
%
where $\rho = \mathbb{E}(\roundbatches)$. Observe that, to guarantee stability, we must have $\mathbb E(\roundbatches) < 1$. Note also that the steady state probabilities up to a given state $j$ can be recursively calculated using the first $j$ equations in \eqref{eq:ssY}.

\subsubsection{Upper and lower bounds on the batch size $\batchsizeq$ under waiting time requirements} \label{sec:bounds}

In this session, we will evaluate the effect of the batch size $Q$ into the waiting time, therefore process $Y_t, \,t \ge 0,$ will be renamed as $Y_t^Q, \, t \ge 0,$ to make the dependence on the batch size $Q$ explicit. 
Let $\omega$ be the waiting time of a patient joining the queue, and assume that the healthcare system requires that:
\begin{equation} \label{eq:waitthreshold}
P( \omega > \tau ) < \conflevelrq,
\end{equation}
for some target waiting time $\waittimereq >0$ and probability threshold $\conflevelrq \in (0,1)$. Now, if a patient arrives and encounters $Y_t^Q > 0$ batches of patients in the queue, their turn will only arrive once these $Y_t^Q$ batches are processed, which will take $Y_t^Q$ time periods. This means that the patient will have to wait $\waittimereq = Y_t^Q$ periods for their turn. Therefore, Eq. \eqref{eq:waitthreshold} can be equivalently stated as:
\begin{equation} \label{eq:waitthresholdy}
\sum_{i = 0}^\waittimereq \pi_{Y^Q} > 1 - \conflevelrq,
\end{equation}
where $\pi_{Y^Q}$ is the steady state distribution of proces $Y_t^Q, \, t \ge 0,$ and satisfies Eq. \eqref{eq:ssY}. Therefore, to limit the waiting time according to the prescribed requirements, we can establish a lower bound ($\lowq$) on the batch size $\batchsizeq > 0$ as follows:
\begin{equation} \label{eq:under}
\lowq = \min \{ \batchsizeq \in \mathbb N : \sum_{i = 0}^\waittimereq \pi_{Y^Q}(i) > 1- \conflevelrq \}.
\end{equation}

When $Y_t^Q = 0$, on the other hand, the next batch will only be processed when the system leaves this state, therefore by applying the strong Markov property \citep{Bremaud2020}, we have:
\[
\omega = \min \{\delta > 0 : Y_{\delta}^Q > 0 | Y_{0}^Q = 0 \}.
\]
Considering the transition matrix in Eq.~\eqref{eq:ydynamicY}-\eqref{eq:transexplicit}, it follows that $\omega$ will be a geometrically distributed variable and:
\begin{equation} \label{eq:upperbound}
P( \omega > \waittimereq | Y_{0}^Q = 0 )  = {k_0(Q)}^\tau,.
\end{equation}
where $k_0(Q)$ is the value of $k_0$ in Eq. \eqref{eq:transexplicit} for a batch size of $Q$. 
One can intuitively see that the value of $k_0(Q)$ will increase as $\batchsizeq$ increases, since in this case we will need more patient arrivals until we are able to process a full batch of patients. Hence, from Eq. \eqref{eq:upperbound}, we can set up an upper bound $\upq$ for the batch size such that:
\begin{equation} \label{eq:overbound}
\upq = \max \{ \batchsizeq \in \mathbb N : k_0(Q)^{\waittimereq} < \conflevelrq \}.
\end{equation}

In the example below, we show how to derive general bounds $\lowq$ and $\upq$ for a given demand distribution and waiting time requirements.

\begin{ex}[Deriving lower and upper bounds for $\batchsizeq$]
Consider an elective surgery queue with monthly demand $\rvdemand$ that satisfies the distribution in Table \ref{tab:demprob}.
\begin{table}[h!]
\centering
\caption{Monthly demand distribution for surgeries} \label{tab:demprob}
\scriptsize{
\begin{tabular}{lccccccccccc} 
\toprule
\textbf{Demand ($\drawnrvdemand$)} & 0 & 1 & 2 & 3 & 4 & 5 & 6 & 7 & 8 & 9 & 10 \\ 
\midrule
$P(\rvdemand=\drawnrvdemand)$ & 0.0308 & 0.0393 & 0.1179 & 0.2096 & 0.2445 & 0.1956 & 0.1086 & 0.0414 & 0.0103 & 0.0015 & 0.0005 \\
\bottomrule
\end{tabular}
}
\end{table}
The expected monthly demand is therefore $\mathbb{E}(\rvdemand) = 3.9022$. Setting $\batchsizeq = 4$, which ensures long-term stability, we have $\statespace_{\stocprocremain{}{}} = \{0, 1, 2, 3\}$. This allows us to evaluate the individual distributions of $\roundbatches(\stocprocremainstate), \, \stocprocremainstate \in \statespace_{\stocprocremain{}}$ as defined in Eq. \eqref{eq:Az}, which we depict in the Table \ref{tab:dembatches}, along with the distribution of variable $\roundbatches$ from Eq. \eqref{eq:batcharr}.

\begin{table}[h!]
\centering 
\caption{Monthly demand distribution for surgeries} 
\label{tab:dembatches}
\scriptsize
\begin{tabular}{lcccc}
\toprule
$P(\roundbatches(\stocprocremainstate) = \drawnrvdemand)$ 
& \multicolumn{4}{c}{\textbf{Batch demand ($\drawnrvdemand$)}} \\ 
\cmidrule{2-5}
& 0 & 1 & 2 & 3 \\
\midrule
$\stocprocremainstate=0$ & 0.3976 & 0.5901 & 0.0123 & 0 \\
$\stocprocremainstate=1$ & 0.1880 & 0.7583 & 0.0537 & 0 \\
$\stocprocremainstate=2$ & 0.0701 & 0.7676 & 0.1618 & 0.0005 \\
$\stocprocremainstate=3$ & 0.0308 & 0.6113 & 0.3559 & 0.002 \\
\textbf{$P(\roundbatches=\drawnrvdemand)$} 
& \textbf{0.171625} & \textbf{0.681825} & \textbf{0.145925} & \textbf{0.000625} \\
\bottomrule
\end{tabular}
\end{table}

The distributions in Table \ref{tab:dembatches} follow directly from Eq. \eqref{eq:Az} and Eq. \eqref{eq:batcharr}. For example, when $\stocprocremainstate=0$, we know that the queue comprises an integer number of batches. Therefore, $P(\roundbatches(0) = 0) = p_{\rvdemand}(0) + p_{\rvdemand}(1) + p_{\rvdemand}(2) + p_{\rvdemand}(3)$ (Table \ref{tab:demprob}) as up to 3 new arrivals will not suffice to generate a new batch of $Q=4$ patients. Analogously, $\roundbatches(0)=1$ if we observe from 5 to 7 arrivals and $\roundbatches(0)=2$ when 8 to 10 patients arrive in the period, as these will suffice for 2 new batches. 

Conversely, when $\stocprocremainstate=3$, we have three patients in excess of a discrete number of batches, which means that $P(\roundbatches(3) = 0) = p_{\rvdemand}(0)$ as new batches will be formed whenever one or more patients arrive. One new batch will form if $\rvdemand \in [1,4]$, two new batches if $\rvdemand \in [5,8]$ and three new batches when $\rvdemand \in [8,10]$. Specifically, when 10 patients arrive, they will form a set of 13 patients when joined to the $\stocprocremainstate=3$ remaining patients from the start of the period. These 3 patients will then suffice to form three batches of $\batchsizeq=4$ patients, with one remaining batchless for the following period.

Note that $\mathbb{E}(\roundbatches) = 0.9755$, and that this value is equal to $\rho = \dfrac{\mathbb{E}(\rvdemand)}{\batchsizeq} = \dfrac{3.9022}{4} = 0.9755$. This is expected, as the dynamics of process $\stocprocbatches{\stocprocperiod}$ describes the original system, only looking at it in terms of batches of patients. Note that $\mathbb{E}(\roundbatches)$ is simply the expected number of batches of $\batchsizeq$ patients arriving and since we process $\mu=1$ batch per time period, we have $\rho = \mathbb{E}(\roundbatches)$ when looking at the evolution of the batch process $\stocprocbatches{\stocprocperiod}, \, \stocprocperiod \ge 0$.

Now, let us suppose that $\waittimereq = 6$ and $\conflevelrq = 0.05$ in~\eqref{eq:waitthreshold}, which means that the system requires 95\% of the patients to be operated in up to six months. To verify whether that is satisfied, we need firstly to solve Eq. \eqref{eq:ssY} for the steady state distribution of $\stocprocbatches{\stocprocperiod}, \, \stocprocperiod \ge 0$ and then use it to verify~\eqref{eq:waitthresholdy}. From Table~\ref{tab:dembatches}, we obtain $k_0=0.171625$,  $k_1=0.681825$, $k_2 = 0.145925$ and $k_3=0.000625$, which can be applied in Eq.~\eqref{eq:ssY} to find $\steadystate{\stocprocbatches{}}(0) = 0.02445$, $\steadystate{\stocprocbatches{}}(1) = 0.1182$, $\steadystate{\stocprocbatches{}}(2) = 0.1219$, $\steadystate{\stocprocbatches{}}(3) = 0.1046$, $\steadystate{\stocprocbatches{}}(4) = 0.0898$, $\steadystate{\stocprocbatches{}}(5) = 0.077$, $\steadystate{\stocprocbatches{}}(6) = 0.0661$. This yields:
\[
P(\omega \le 6 | \stocprocbatches{\stocprocperiod} \ge 1) \approx 0.602 < 0.95 = 1 - \conflevelrq.
\]
This implies that $\underline Q > 4$ in Eq. \eqref{eq:under} and we should, therefore, recalculate the probability above for increasing values of $Q$ until we find the first value that satisfies the waiting time criterion - which will give us the lower bound $\underline Q$.

To verify the waiting time requirements when $\stocprocbatches{\stocprocperiod} = 0$, one can apply \eqref{eq:upperbound}, which implies that $P(\omega \le 6 | \stocprocbatches{\stocprocperiod} = 0) = 0.171625^6 \approx 2.5 \cdot 10^{-5} < 0.05$. Therefore, to find the bound $\overline Q$ in Eq. \eqref{eq:overbound}, we must iterate on higher values of $Q$.

When verifying Eq.~\eqref{eq:under} with $Q = 10$, it implies $P(\omega \le 6 | \stocprocbatches{\stocprocperiod} = 0) = 0.609780^6 \approx 0.05114 > 0.05$. Then, if one set $\batchsizeq = 9$, it implies $P(\omega \le 6 | \stocprocbatches{\stocprocperiod} = 0) = 0.566478^6 \approx 0.03304 < 0.05$. Thus, we can set $\upq = 9$.

For the $\lowq$, one can verify $\batchsizeq = 5$, which implies $k_0=0.265720$,  $k_1=0.688120$, $k_2 = 0.046160$ and $k_3=0$, which can be applied in Eq. \eqref{eq:ssY} to find $\steadystate{\stocprocbatches{}}(0) = 0.21956$, $\steadystate{\stocprocbatches{}}(1) = 0.606723$, $\steadystate{\stocprocbatches{}}(2) = 0.143539$, $\steadystate{\stocprocbatches{}}(3) = 0.024935$, $\steadystate{\stocprocbatches{}}(4) = 0.004332$, $\steadystate{\stocprocbatches{}}(5) = 0.000752$, $\steadystate{\stocprocbatches{}}(6) = 0.000131$. This yields:
\[
P(\omega \le 6 | \stocprocbatches{\stocprocperiod} \ge 1) \approx 0.999973 > 0.95 = 1 - \conflevelrq.
\]
Therefore, we can set $\lowq = 5$.

Given the bounds $\lowq = 5$ and $\upq = 9$, any value of $\batchsizeq$ satisfying $\lowq \le \batchsizeq \le \upq$ can be selected. 
This ensures that the waiting-time requirements are met both for patients arriving to an empty queue and for those arriving when batches of patients are already waiting.

\end{ex}
%
\subsection{Module 2: Modelling Uncertainty in Surgery Duration and Block Sessions} \label{sec:module2}

This section details the components of module 2 of the proposed framework, see Section \ref{sec:proposal}. 
Firstly, it is necessary to estimate how many elective patients can be accommodated within a single OT block for each subspeciality, to ensure consistency between the periodic schedule and the waiting-time requirements established by the $(R,Q)$ policy in Section~\ref{sec:waitinglist}. 
This is key to ensuring that the number of sessions allocated to a given subspeciality is enough to warrant a number of surgeries that lies in the corresponding interval $[\lowq, \upq]$ that ensures compliance with prescribed waiting times, considering the probabilistic distribution of surgery durations. 
Specifically, a trade-off arises between idle time, that results in underutilised theatre capacity, and overtime, which may trigger cancellations of scheduled procedures. 



First, we introduce a Newsvendor-based model to guide decisions on whether to perform or cancel the next scheduled surgery, depending on the time remaining in the session. 
Then, assuming that these decisions are taken sequentially at the end of each scheduled surgery, we introduce an inventory-based model to find the number of patients that can be reliably scheduled within a time block, as well as the probability distribution of the number of cancellations at each session. 
These are inputs required for subsequent tactical OT allocation decisions.


\subsubsection{Managing surgery cancellations via a Newsvendor-based model\label{sec:cancel}}
Deciding whether to perform the next scheduled surgery in a session or cancel it due to limited remaining time directly affects both OT utilisation and overtime risk. 
Performing the surgery with insufficient time may result in overtime and disrupt subsequent scheduling decisions, whereas cancelling it despite sufficient available time may lead to underutilisation of OT capacity and reduced patient access to care.
To support this decision, we propose a Newsvendor-based model that determines an optimal time budget for the next surgery, i.e., the minimum amount of time that must remain in the session to greenlight the surgery.
The symbols used in this section are defined in Table~\ref{tab:symb_newsvendor}.

\begin{table}[htb]
    \centering
    \caption{Symbols used in the Newsvendor-based model.\label{tab:symb_newsvendor}}
    \scriptsize
    \begin{tabular}{cl}
    \toprule
    \textbf{Set} & \textbf{Description} \\
    \midrule
    $\specialityset$ & Set of specialities \\
    $\supportnv$ & Support of the surgery-duration random variable \\
    \midrule
    
    \textbf{Index} & \textbf{Description} \\
    \midrule
    $\specialityidx$ & speciality index \\
    \midrule

    \textbf{Parameter} & \textbf{Description} \\
    \midrule
    $\blocklength$ & Length of an OT block \\
    $t_b$ & Time budget allocated to the last surgery in a block \\
    $\excesscost$ & Unit cost of unused allocated time \\
    $\overtimepenalty$ & Unit penalty cost of overtime \\
    $\safetythreshold$ & Safety time threshold for starting the next surgery \\
    \midrule

    \textbf{Random Variable} & \textbf{Description} \\
    \midrule
    $\rvsurgdur$ & Random duration of a surgery \\
    $\beta_1$ & Unused allocated time \\
    $\beta_2$ & Overtime duration \\
    \midrule

    \textbf{Function} & \textbf{Description} \\
    \midrule
    $f_{\rvsurgdur}(\cdot)$ & Probability mass function of $\rvsurgdur$ \\
    $F_{\rvsurgdur}(\cdot)$ & Cumulative distribution function of $\rvsurgdur$ \\
    $u(t_b)$ & Expected cost of allocating time budget $t_b$ \\
    \midrule
    \end{tabular}
\end{table}

Let $\rvsurgdur$ be a non-negative random variable taking values from $\supportnv \subseteq \mathbb R_+$, which represents the duration of a given surgery from subspeciality $\specialityidx \in \specialityset$, with cumulative distribution $F_{\rvsurgdur}: \supportnv \to [0,1]$. Assume that, once the theatre is ready for the surgery, it only proceeds as planned if the remaining time in the session exceeds a prescribed budget of $t_{b,\specialityidx} \in \supportnv$ time units. Otherwise, the surgery is called off and the session is terminated.

The time budget $t_{b,\specialityidx} \in \supportnv$ can be interpreted as a perishable inventory of session time of subspeciality $\specialityidx \in \specialityset$. If the inventory is excessive, we will not fully utilise the time allocated to the surgical team to perform the surgery. Conversely, if we proceed with the surgery but the time budget is insufficient, the surgical team will have to work overtime to complete the surgery. Let $\excesscost > 0$ be the cost of each unit of time budget in excess of the actual surgical time $\rvsurgdur$ (i.e., idle time cost) and $\overtimepenalty > \excesscost$ be a penalty for each unit of overtime that the surgical team incurs. The expected cost of the time deviation between $\rvsurgdur$ and $t_{b,\specialityidx}$ is defined as: 
\begin{equation} \label{eq:newsvendor}
u( t_{b,\specialityidx} ) = \excesscost\,\mathbb{E} (\beta_1) + \overtimepenalty\,\mathbb{E}( \beta_2 ), 
\end{equation}
where $\beta_1 = \min\{ t_{b,\specialityidx} - \rvsurgdur, \; 0 \}$ is the underused time budget, and $\beta_2 = \max \{ \rvsurgdur - t_{b,\specialityidx}, \; 0 \}$ is the total overtime. The objective is to find an optimal time budget allocation $\safetythreshold$ that minimizes the time deviation cost in Eq. \eqref{eq:newsvendor}, which is given by:
\begin{equation} \label{eq:mintimebudget}
    \safetythreshold =  \arg \min_{t_{b,\specialityidx} \in \supportnv} u( t_{b,\specialityidx} ).
\end{equation}

\begin{lemma}
    Let $\beta_1 = \min\{ t_{b,\specialityidx} - \rvsurgdur, \; 0 \}$ and $\beta_2 = \max \{ \rvsurgdur - t_{b,\specialityidx}, \; 0 \}$. Then, it follows that:
\begin{equation} \label{eq:beta}
 \mathbb{E}(\beta_1 ) =  \int_0^{t_{b,\specialityidx}} F_{\rvsurgdur}(t) dt, \quad \text{and} \quad    \mathbb{E}(\beta_2 ) = \int_{t_{b,\specialityidx}}^\infty (1 - F_{\rvsurgdur}(t) ) dt.
\end{equation}
\end{lemma}
\begin{proof}
We will start by proving the result for $\beta_1$. Since $\beta_1 \ge 0$, it follows that
\begin{align*}
\mathbb{E}(\beta_1) = \int_{t=0}^\infty P ( \beta_1 > t ) dt = \int_{t=0}^{t_{b,\specialityidx}} P( t_{b,\specialityidx} - \rvsurgdur > t ) dt 
=  \int_{t=0}^{t_{b,\specialityidx}} P( \rvsurgdur < t_{b,\specialityidx} - t ) dt = \int_{t=0}^{t_{b,\specialityidx}} F_{\rvsurgdur}(t_{b,\specialityidx} - t) dt 
= \int_{t=0}^{t_{b,\specialityidx}} F_{\rvsurgdur}(t) dt.
\end{align*}

Since $\beta_2$ is also non-negative, we can write
\begin{align*}
    \mathbb{E}(\beta_2) =  \int_0^\infty P( \rvsurgdur - t_{b,\specialityidx} > t) dt = \int_0^\infty P( \rvsurgdur  > t + t_{b,\specialityidx}) dt
    = \int_{t_{b,\specialityidx}}^\infty P( \rvsurgdur  > t ) dt = \int_{t_{b,\specialityidx}}^\infty (1 - F_{\rvsurgdur}(t) ) dt.
\end{align*}
\end{proof}

\begin{theo} \label{thm:quantile}
Let $u(t_{b,\specialityidx})$ be defined as in Eq.~\eqref{eq:newsvendor}, and $\safetythreshold$ be a solution to Eq. \eqref{eq:mintimebudget}. Then, if follows that
\begin{equation} \label{eq:optimalquantile}
\safetythreshold = F_{\rvsurgdur}^{-1} \left( \dfrac{ \dfrac{\overtimepenalty}{\excesscost}}{1 + \dfrac{\overtimepenalty}{\excesscost}}\right),
\end{equation}
where $F_{\rvsurgdur}^{-1}(\cdot)$ is the inverse of the cumulative distribution function $F_{\rvsurgdur}(\cdot)$.
\end{theo}
\begin{proof}
    Substituting \eqref{eq:beta} in Eq. \eqref{eq:newsvendor} yields
\[
u( t_{b,\specialityidx} ) = \excesscost  \int_0^{t_{b,\specialityidx}} F_{\rvsurgdur}(t) dt + \overtimepenalty \int_{t_{b,\specialityidx}}^\infty (1 - F_{\rvsurgdur}(t) ) dt.
\]
When $t_{b,\specialityidx} = \safetythreshold$, the first order optimality condition for $u(t_{b,\specialityidx})$ implies
\[
\dfrac{d u( t_{b,\specialityidx} )}{d t_{b,\specialityidx}} = \excesscost F_{\rvsurgdur}(\safetythreshold) - \overtimepenalty ( 1 - F_{\rvsurgdur}(\safetythreshold) ) = 0.
\]
Hence, we must have
\[
\excesscost F_{\rvsurgdur}(\safetythreshold) - \overtimepenalty ( 1 - F_{\rvsurgdur}(\safetythreshold) ) = 0 \implies \dfrac{F_{\rvsurgdur}(\safetythreshold)}{( 1 - F_{\rvsurgdur}(\safetythreshold) )} = \dfrac{\overtimepenalty}{\excesscost} \implies F_{\rvsurgdur}(\safetythreshold) \left( 1 + \dfrac{\overtimepenalty}{\excesscost} \right) = \dfrac{\overtimepenalty}{\excesscost}.
\]
We conclude the proof by noting that \eqref{eq:optimalquantile} follows from the expression above.
\end{proof}

Theorem \ref{thm:quantile} implies that, whenever the remaining time in a session exceeds the quantile $\safetythreshold$ in Eq. \eqref{eq:optimalquantile}, the surgery is greenlighted to continue. Otherwise, the surgery is cancelled and the session comes to a close. Note that the quantile is solely a function of the ratio between the overtime cost $\overtimepenalty$ and the cost of unused session time $\excesscost$, and is applicable to any surgery time distribution. For example, if $\overtimepenalty = 2\excesscost$, $\safetythreshold = F^{-1}_{\rvsurgdur} \left( \dfrac{2}{3} \right)$. This means that at least two thirds of the greenlighted surgeries will not require overtime.

\subsubsection{Applying the Newsvendor-based model \label{sec:applicationnewsvendor}}
Building on the theoretical results of the previous subsection, we estimate, for each surgical subspeciality, the distribution of the number of surgeries that can be completed within a block of length $\blocklength \in \mathbb{Z}_{+}$. 
This distribution is then used to evaluate session-planning decisions under an inventory-based formulation.
For this purpose, we use the notation introduced in Table~\ref{tab:symb_appnewsvendor}. 

\begin{table}[htb]
    \centering
    \caption{Symbols used in the derivation of the completion distribution.\label{tab:symb_appnewsvendor}}
    \scriptsize
    \begin{tabular}{cl}
    \toprule
    \textbf{Set / Index / Parameter} & \textbf{Description} \\
    \midrule
    $\specialityset$ & Set of surgical subspecialities \\
    $\specialityidx$ & Index of a surgical subspeciality, such that $\specialityidx \in \specialityset$ \\
    $\blocklength$ & Length of the operating block \\
    $\safetythreshold$ & Safety time threshold that must remain available in the block to accommodate one further surgery \\
    $\probstoppingthreshold$ & Probability threshold used to truncate the tail of the distribution \\
    $\kmax$ & Largest value of $k$ such that $P(\numcompletedsurgs > k) > \probstoppingthreshold$, for $\specialityidx \in \specialityset$ \\
    \midrule

    \textbf{Random Variable} & \textbf{Description} \\
    \midrule
    $\rvsurgdur$ & Surgery duration for a given subspeciality $\specialityidx \in \specialityset$ \\
    $\cumulativeduration{k}$ & Cumulative duration of the first $k$ surgeries, $\cumulativeduration{k}=\sum_{i=1}^{k}\rvsurgdur^{(i)}$, for $\specialityidx \in \specialityset$\\
    $\numcompletedsurgs$ & Number of surgeries that can be completed within a block, for $\specialityidx \in \specialityset$ \\
    \midrule

    \textbf{Function / Derived Quantity} & \textbf{Description} \\
    \midrule
    $f_{\rvsurgdur}$ & Probability mass function of the duration of a single surgery \\
    $F_{\rvsurgdur}$ & Cumulative distribution function of the duration of a single surgery \\
    $f_{\cumulativeduration{k}}$ & Probability mass function of $\cumulativeduration{k}$ \\
    $F_{\cumulativeduration{k}}$ & Cumulative distribution function of $\cumulativeduration{k}$ \\
    $f_{\numcompletedsurgs}$ & Probability mass function of $\numcompletedsurgs$ \\
    $F_{\numcompletedsurgs}$ & Cumulative distribution function of $\numcompletedsurgs$ \\
    $\ast$ & Convolution operator \\
    $\survivalvector$ & Vector with entries $\survivalvector[k]=P(\numcompletedsurgs > k)$ \\
    \bottomrule
    \end{tabular}
\end{table}

For any $k \in \mathbb{N}_{+}$ and $\specialityidx \in \specialityset$, the cumulative duration of the first $k$ surgeries is
\[
\cumulativeduration{k}=\sum_{i=1}^{k}\rvsurgdur^{(i)},
\]
with $f_{\cumulativeduration{k}} : \supportnv \to [0,1]$ and $F_{\cumulativeduration{k}} : \supportnv \to [0,1]$ denoting the probability mass function and cumulative distribution function of $\cumulativeduration{k}$, respectively.

Let $\numcompletedsurgs$ be the random variable denoting the number of surgeries of subspeciality $\specialityidx \in \specialityset$ that can be completed within a block of length $\blocklength$.
After $\completedsurgsidx \in \mathbb{N}_+$ surgeries have been completed, one further surgery can be accommodated only if the residual time in the block exceeds the safety threshold $\safetythreshold$. 
Therefore,
\begin{equation}\label{eq:survival_k_revised}
P(\numcompletedsurgs > \completedsurgsidx)
=
P(\blocklength-\cumulativeduration{\completedsurgsidx}>\safetythreshold)
=
P(\cumulativeduration{\completedsurgsidx}<\blocklength-\safetythreshold).
\end{equation}
Thus, the event $\{\numcompletedsurgs > \completedsurgsidx\}$ is equivalent to the event that, after completing $\completedsurgsidx$ surgeries, sufficient time remains in the block to start one additional surgery while preserving the required safety threshold $\safetythreshold$.

Because $\cumulativeduration{\completedsurgsidx}$ is the sum of $\completedsurgsidx$ independent surgery durations, its distribution can be obtained by the following recurrence, where $\ast$ denotes the convolution operator:
\[
f_{\cumulativeduration{\completedsurgsidx}} =
\begin{cases}
f_{\rvsurgdur}, 
& \text{if } \completedsurgsidx = 1, \\[4pt]
f_{\cumulativeduration{\completedsurgsidx-1}} \ast f_{\rvsurgdur},
& \text{otherwise}.
\end{cases}
\]

For each $\completedsurgsidx$, Eq.~\eqref{eq:survival_k_revised} then yields the survival probability $P(\numcompletedsurgs > \completedsurgsidx)$.
These survival probabilities are computed sequentially in Algorithm \ref{alg:newsvendor} and stored in the vector $\survivalvector$, for each $\specialityidx \in \specialityset$, whose $\completedsurgsidx$-th entry is defined as
\[
\survivalvector[\completedsurgsidx]=P(\numcompletedsurgs > \completedsurgsidx).
\]
The procedure starts at $\completedsurgsidx=0$, with $\survivalvector[0]=1$ (i.e., the first surgery is always completed), and stops when $\survivalvector[\completedsurgsidx]\leq \probstoppingthreshold$ to preserve relevant support of the distribution while avoiding unnecessary convolution steps in the negligible tail. 
For each subspeciality, the algorithm therefore returns both the vector $\survivalvector$ and the largest relevant index $k$, defined as $\kmax$.

\begin{algorithm}[H]
\scriptsize{
\caption{Computation of $P(\numcompletedsurgs > \completedsurgsidx)$ via convolutions\label{alg:newsvendor}}
\begin{algorithmic}[1]
\State \textbf{Input:} $\blocklength$, $\safetythreshold$, $\probstoppingthreshold$, subspecialities $\specialityidx \in \specialityset$, and duration distributions $f_{\rvsurgdur}$ for each $\specialityidx \in \specialityset$
\State \textbf{Output:} vectors $\survivalvector$, $\specialityidx \in \specialityset$, such that $\survivalvector[\completedsurgsidx]=P(\numcompletedsurgs > \completedsurgsidx)$; and $\kmax$ values for each $\specialityidx \in \specialityset$
\For{each $\specialityidx \in \specialityset$}
    \State Initialise $\survivalvector[0] \gets 1$, $\completedsurgsidx \gets 1$, and $f_{\cumulativeduration{1}} \gets f_{\rvsurgdur}$
    \While{$\survivalvector[\completedsurgsidx-1] > \probstoppingthreshold$}
        \If{$\completedsurgsidx > 1$}
            \State $f_{\cumulativeduration{\completedsurgsidx}} \gets f_{\cumulativeduration{\completedsurgsidx-1}} \ast f_{\rvsurgdur}$
        \EndIf
        \State $\survivalvector[\completedsurgsidx] \gets P(\cumulativeduration{\completedsurgsidx} < \blocklength - \safetythreshold)$
        \State $\completedsurgsidx \gets \completedsurgsidx + 1$
    \EndWhile
    \State $\kmax \gets \completedsurgsidx - 1$
\EndFor
\end{algorithmic}
}
\end{algorithm}

Once the survival function has been obtained, the probability distribution of $\numcompletedsurgs$ follows from the standard relationship between the probability mass function and the survival function of a discrete random variable, that is, for $\completedsurgsidx \geq 1$,
\begin{equation}\label{eq:pmf_numcompletedsurgs_revised}
P(\numcompletedsurgs = \completedsurgsidx)
=
P(\numcompletedsurgs > \completedsurgsidx-1)-P(\numcompletedsurgs > \completedsurgsidx).
\end{equation}

Hence, the convolution procedure yields the probability distribution of the number of surgeries that can be completed within the block for each subspeciality. 
This distribution is the main output of the present subsection and constitutes the probabilistic input to the inventory-based planning model described next.

\subsubsection{Inventory-based model for defining the number of surgeries in a time block\label{sec:patientsperblock}}

Using the results derived in the previous subsection, we now determine the number of surgeries to be planned within a single OT block. 
This decision is formulated through an inventory-based~\citep{Axsater2015Inv} perspective that balances the benefit of scheduling additional surgeries against the expected cost of cancellations.
The symbols used in this subsection are defined in Table~\ref{tab:symb_inventory}.

\begin{table}[htb]
    \centering
    \caption{Symbols used in the inventory-based planning model.\label{tab:symb_inventory}}
    \scriptsize
    \begin{tabular}{cl}
    \toprule
    \textbf{Parameter / Decision Variable} & \textbf{Description} \\
    \midrule
    $\plannedsurgs$ & Number of surgeries planned in a block \\
    $\optplannedsurgs$ & Optimal number of surgeries planned in a block \\
    $\cancelledsurgeries$ & Number of cancelled surgeries \\
    $\marginalvalue$ & Marginal value associated with scheduling one additional surgery \\
    \midrule

    \textbf{Function / Derived Quantity} & \textbf{Description} \\
    \midrule
    $\gain{\plannedsurgs}$ & Total gain associated with planning $\plannedsurgs$ surgeries \\
    $\cancellationfunction(\cancelledsurgeries)$ & Cost incurred when $\cancelledsurgeries$ surgeries are cancelled \\
    $\cancellationcost(\plannedsurgs)$ & Expected cancellation cost when $\plannedsurgs$ surgeries are planned \\
    $\Delta(\plannedsurgs)$ & Net value of planning $\plannedsurgs$ surgeries \\
    \bottomrule
    \end{tabular}
\end{table}

Let $\plannedsurgs \in \mathbb{N}_{+}$ denote the number of surgeries planned in a block. 
Following the previous definition that the random number of surgeries that can indeed be completed is $\numcompletedsurgs$, and assuming that $\numcompletedsurgs \leq \plannedsurgs$, the number of cancellations is given by
\[
\cancelledsurgeries = \plannedsurgs-\numcompletedsurgs.
\]

Hence, for $\cancelledsurgeries = 0,1,\dots,\plannedsurgs$, the probability of cancelling $\cancelledsurgeries$ surgeries when $\plannedsurgs$ surgeries are planned is given by
\begin{equation}\label{eq:cancellation_prob}
P\bigl(\numcompletedsurgs = \plannedsurgs-\cancelledsurgeries\bigr),
\qquad \cancelledsurgeries = 0, 1,\dots,\plannedsurgs.
\end{equation}

We next associate each planning decision with an expected total gain and an expected cancellation cost. 
We define the expected total gain of adding an $\plannedsurgs-$th surgery as:
\begin{equation}\label{eq:gain_function_revised}
    \gain{\plannedsurgs} = \marginalvalue \cdot \sum_{j=0}^n P(\numcompletedsurgs \geq j),
\end{equation}
where $\marginalvalue>0$ is a constant representing the marginal value of scheduling one additional surgery. 
Eq.~\eqref{eq:gain_function_revised} is directly obtained from the survival probabilities computed in the previous subsection, since $P(\numcompletedsurgs \geq j)=P(\numcompletedsurgs > j-1)$.

In this setting, $\gain{\plannedsurgs}$ is not interpreted as a direct monetary return, but rather as the total benefit associated with the number of planned surgeries. 
This interpretation is particularly suitable in public healthcare settings, where the value of additional surgeries may reflect improved access to care, reduced waiting times, and better utilisation of OT capacity.

The cost associated with cancellations is modelled through a function $\cancellationfunction:\mathbb{Z}_{+} \rightarrow\mathbb{R}_{+}$ satisfying $\cancellationfunction(0) = 0$ and assumed to be non-decreasing. 
The expected cancellation cost when planning $\plannedsurgs$ surgeries is therefore
\begin{equation}\label{eq:cost_function_revised}
    \cancellationcost(\plannedsurgs)
    =
    \mathbb{E}\bigl[\cancellationfunction\bigl(\plannedsurgs-\numcompletedsurgs)]
    =
    \sum_{k=0}^{n} P(\numcompletedsurgs=k)\, \cdot \cancellationfunction\bigl(\plannedsurgs-k).
\end{equation}

The analogous net value associated with planning $\plannedsurgs$ surgeries is then defined as
\begin{equation}\label{eq:net_value_revised}
    \Delta(\plannedsurgs)=\gain{\plannedsurgs}-\cancellationcost(\plannedsurgs).
\end{equation}

Because $\gain{\plannedsurgs}$ represents the overall benefit scheduling $\plannedsurgs$ surgeries, the optimal planning level is the largest value for which the net value remains positive, i.e.,
\begin{equation}\label{eq:optimal_planned_surgeries}
    \optplannedsurgs=\max\{\plannedsurgs \in \mathbb{N}_{+} : \Delta(\plannedsurgs)>0\}.
\end{equation}
Thus, $\optplannedsurgs$ identifies the last point at which the total gain of scheduling a given number of surgeries exceeds the expected cost of cancellations.

Once $\optplannedsurgs$ has been determined, the relevant probabilistic input for the subsequent OT allocation model is the distribution of the number of surgeries completed under that planning decision, i.e.,
\begin{equation} \label{eq:distnew}
P(\numcompletedsurgs = \completedsurgsidx), \qquad \completedsurgsidx=0,1,\dots,\optplannedsurgs,
\end{equation}
together with the implied probability distribution of the number of cancellations. These quantities are then passed to the next stage of the framework.
%
\subsection{Module 3: Solving an OT Scheduling Problem\label{sec:ot}}

Following the proposed integrated framework, the final step is to solve an OTSP, which addresses the tactical allocation of surgical subspecialities to OT blocks. 
This stage translates the strategic decisions derived in the preceding modules into implementable schedules for hospital operations. 
In particular, the OTSP takes as input the target number of surgeries for each subspeciality, as determined by the $(R,Q)$ policy in the \rqmodule~module, together with the effective block capacity obtained from the \overtimemodule~module.

We formulate the OTSP as a Mixed-Integer Linear Programming (MILP) model. 
This formulation allows the explicit representation of practical constraints and multiple, potentially competing objectives, while ensuring optimality with respect to the specified criteria. 
The model is designed to reflect hospital priorities such as high OT utilisation, the allocation of priority subspecialities, and the efficient sequencing of blocks from a logistical perspective.

The proposed framework is, however, not restricted to the MILP approach adopted here. 
Since the previous modules provide demand and effective capacity as structured inputs, alternative solution methods, including heuristics and metaheuristics, could also be employed in the OTSP stage.
The main value of the framework therefore lies in the integration of decision layers, rather than in the exclusive use of a particular scheduling technique.

In the following subsections, we first describe how the strategic outputs of the $(R,Q)$ policy and the effective capacity information are preprocessed into block-demand units, which constitute the relevant unit of the scheduling problem. 
We then present the MILP formulation of the OTSP, including its decision variables, objective functions, and constraints.

\subsubsection{Preprocessing Strategic Input\label{subsec:preprocess_otsp}}

The OTSP is formulated at the block level. 
Therefore, the strategic outputs of the previous modules must be converted into the number of blocks required for each subspeciality.

For each subspeciality $\specialityidx \in \specialityset$, let
\[
\probdistplan(k) = P(\numcompletedsurgs = k), \qquad k=0,1,\dots,\optplannedsurgs_{\specialityidx},
\]
denote the probability distribution of the number of surgeries that can be completed in a single block under the optimal planning decision obtained in the \overtimemodule~module.
This distribution summarises the effective stochastic capacity of one block for subspeciality $\specialityidx$.

To determine, for each $\specialityidx \in \specialityset$, the number of blocks required to achieve $(R,Q)$ policy targets $\lowq_{\specialityidx}$ and $\upq_{\specialityidx}$ at confidence level $\conflevelot \in [0,1]$, we use successive convolutions of $\probdistplan(\cdot)$. 
Let
\[
\probdistplan^{(n)} = \underbrace{\probdistplan \ast \cdots \ast \probdistplan}_{n \text{ times}}
\]
denote the probability distribution of the total number of surgeries that can be completed across $n$ blocks allocated to subspeciality $\specialityidx$, and let $\cumulpdp^{(n)}$ be its corresponding cumulative distribution function.

Then, for each target $\targetq \in \{\lowq_{\specialityidx},\upq_{\specialityidx}\}$, and for each $\specialityidx \in \specialityset$, the minimum number of blocks required is the smallest integer $n \in \mathbb{N}_{+}$ such that
\begin{equation}\label{eq:block_requirement}
P\!\left(\sum_{b=1}^{n} \numcompletedsurgs^{(b)} \geq \targetq\right)
=
1 - \cumulpdp^{(n)}(\targetq - 1)
\geq
\conflevelot,
\end{equation}
where $\numcompletedsurgs^{(b)}$ extends the random variable $\numcompletedsurgs$ and denotes random variable the number of surgeries completed in the $b$-th block, assumed independent and identically distributed according to $\probdistplan(\cdot)$.

In other words, we seek the smallest number of blocks such that the probability of completing at least $\targetq$ surgeries is no smaller than $\conflevelot$. 
This procedure is applied separately to the lower and upper activity targets. Algorithm~\ref{alg:blockdemand} summarises the computation.

\begin{algorithm}[H]
\scriptsize{
\caption{Block demand estimation via successive convolutions\label{alg:blockdemand}}
\begin{algorithmic}[1]
\State \textbf{Input:} set of subspecialities $\specialityset$, targets $(\lowq_{\specialityidx},\upq_{\specialityidx})$, confidence level $\conflevelot$, $\probdistplan(\cdot)$ probabilities
\State \textbf{Output:} block requirements $(\lowerblockdemand,\upperblockdemand)$ for each $\specialityidx \in \specialityset$

\For{each subspeciality $\specialityidx \in \specialityset$}
    \If{$\lowq_{\specialityidx} = 0$ and $\upq_{\specialityidx} = 0$}
        \State $\lowerblockdemand \gets 0$, $\upperblockdemand \gets 0$
    \Else
        \For{each $\targetq \in \{\lowq_{\specialityidx},\upq_{\specialityidx}\}$}
            \State $n \gets 1$
            \State $\probdistplan^{(n)} \gets \probdistplan$
            \State $\cumulpdp^{(n)} \gets \calccdf\!\left(\probdistplan^{(n)}\right)$
            \While{$1 - \cumulpdp^{(n)}(\targetq - 1) < \conflevelot$}
                \State $\probdistplan^{(n+1)} \gets \convolve\!\left(\probdistplan^{(n)}, \probdistplan\right)$
                \State $n \gets n + 1$
                \State $\cumulpdp^{(n)} \gets \calccdf\!\left(\probdistplan^{(n)}\right)$
            \EndWhile
            \If{$\targetq = \lowq_{\specialityidx}$}
                \State $\lowerblockdemand \gets n$
            \Else
                \State $\upperblockdemand \gets n$
            \EndIf
        \EndFor
    \EndIf
\EndFor
\end{algorithmic}
}
\end{algorithm}

This preprocessing step provides a mechanism for propagating the uncertainty captured in the strategic layer into the tactical scheduling stage without explicitly formulating the OTSP as a stochastic optimisation model. 
As a result, the scheduling problem can be solved as a deterministic MILP while remaining consistent with surgery allocation targets that reflect both variability in surgery durations and the effective capacity of each block, while ensuring prescribed long-term waiting time targets across all surgical subspecialities.

\subsubsection{MILP Formulation\label{subsec:milp}}

Modelling the OTSP as a MILP allows us to capture, in a structured manner, the interactions between subspeciality demand, OT block capacity, and scheduling requirements. 
Throughout the model, we use the notation defined in Table~\ref{tab:symb}.

Let $\roomset$ denote the set of operating theatres, $\planninghorizonset$ the set of days in the planning horizon, and $\timeperiodset$ the set of shift periods. 
The set of blocks is defined as $\blockset=\planninghorizonset\times\timeperiodset$, so that each block corresponds to a day--period pair. 
Let $\specialityset$ denote the set of surgical subspecialities, with $\specialitymset \subseteq \specialityset$, $\specialityprset \subseteq \specialityset$, and $\specialitytxset \subseteq \specialityset$ denoting, respectively, the subspecialities that require two sequential blocks, priority subspecialities, and organ transplant subspecialities. 
For each day $\planninghorizonidx \in \planninghorizonset$, let $\blockidx_1(\planninghorizonidx)$ and $\blockidx_2(\planninghorizonidx)$ denote the morning and afternoon blocks of that day, respectively. 

Let $\blockvar \in \{0,1\}$ indicate whether theatre $\roomidx \in \roomset$, in block $\blockidx \in \blockset$, is allocated to subspeciality $\specialityidx \in \specialityset$. 
Variable $\precvar \in \{0,1\}$ indicates whether theatre $\roomidx \in \roomset$, on day $\planninghorizonidx \in \planninghorizonset$, is allocated exclusively to subspeciality $\specialityidx \in \specialityset$ in the sequential blocks of that day. 
Variable $\integervar \in \mathbb{Z}_+$ denotes the number of transplant pairs allocated in block $\blockidx \in \blockset$ to subspeciality $\specialityidx \in \specialitytxset$. 
Finally, $\teamslackvar \in \mathbb{Z}_+$ and $\anestslackvar \in \mathbb{Z}_+$ denote, respectively, the slack variable for the number of teams available in block $\blockidx \in \blockset$ for subspeciality $\specialityidx \in \specialityset$, and the slack variable for the number of anaesthetists available in block $\blockidx \in \blockset$.

\begin{table}[htb]
    \centering
    \caption{\label{tab:symb}Table of Symbols}
    \scriptsize
    \begin{tabular}{cl}
    \toprule
    \textbf{Set}                    & \textbf{Description}                                                                       \\
    \midrule
        $\roomset$                       & OTs (both elective and outpatient)                       \\
        $\planninghorizonset$            & Days of the planning horizon, such that $\planninghorizonset = \{1, 2, ..., \days\}$  \\
        $\timeperiodset$                 & Shift periods, such that $\timeperiodset = \{\text{morning}, \text{afternoon}\}$          \\
        $\blockset$                      & Blocks, such that $\blockset = \planninghorizonset \times \timeperiodset$     \\
        $\blockmset$                     & Subset of $\blockset$ that includes only blocks in which at least one subspeciality requires two sequential allocated blocks, such that $\blockmset \subset \blockset$ \\
        $\specialityset$                  & Surgical subspecialities                                                                \\
        $\specialitymset$                 & Surgical subspecialities that necessarily require two sequential allocated blocks, such that $\specialitymset \subset \specialityset$          \\
        $\specialityprset$                & Priority surgical subspecialities, such that $\specialityprset \subset \specialityset$  \\
        $\specialitytxset$                & Organ transplant surgical subspecialities, such that $\specialitytxset \subset \specialityset$                                                                                                              \\
    \midrule
    \textbf{Index}                           & \textbf{Description}                                                                  \\
    \midrule
        $\roomidx$ & OT, with $\roomidx \in \roomset$ \\ 
        $\blockidx$ & Block, with $\blockidx \in \blockset$ \\ 
        $\planninghorizonidx$ & Day of the planning horizon, with $\planninghorizonidx \in \planninghorizonset$\\
        $\blockidx_1(\planninghorizonidx),\blockidx_2(\planninghorizonidx)$ & Morning and afternoon blocks $\blockidx \in \blockset$ associated with day $\planninghorizonidx\in\planninghorizonset$ \\
        $\specialityidx$ & Surgical subspeciality, with $\specialityidx \in \specialityset$ \\
    \midrule
    \textbf{Parameter}                        & \textbf{Description}                                                                  \\
    \midrule
        $\upperblockdemand$                                         & Maximum number of blocks allocated to each subspeciality $\specialityidx \in \specialityset$            \\
        $\lowerblockdemand$                                         & Minimum number of blocks allocated to each subspeciality $\specialityidx \in \specialityset$            \\
        $\specteams_{\blockidx,\specialityidx}$              & Number of teams of subspeciality $\specialityidx \in \specialityset$ available in block  $\blockidx \in \blockset$ \\
        $\numbersanests_{\blockidx}$                       & Number of anaesthetists available in block $\blockidx \in \blockset$                         \\
        $\needanest_{\specialityidx}$                       & 1 if subspeciality $\specialityidx \in \specialityset$ requires an anaesthetist; 0 otherwise         \\    
        $\anestconstant$                                   & Constant to penalize the anaesthetist slack variable $\anestslackvar$ in the objective function ($\anestconstant = |\blockset|\cdot|\specialityset|\cdot|\roomset|$) \\
        $\allowedrooms_{\roomidx,\blockidx,\specialityidx}$   & 1 if theatre $\roomidx \in \roomset$, in block $\blockidx \in \blockset$, is preferred to be by subspeciality $\specialityidx \in \specialityset$ in the shift block $\blockidx \in \blockset$; 0 otherwise\\
    \bottomrule
    \textbf{Variable}                         & \textbf{Description}                                                                  \\
    \midrule
        $\blockvar$                                        & 1 if theatre $\roomidx \in \roomset$, in block $\blockidx \in \blockset$, is allocated to subspeciality $\specialityidx \in \specialityset$; 0 otherwise. \\
                                                           & This variable is only created if the allocation is feasible (i.e., speciality $\specialityidx$ is compatible with the theatre $\roomidx$).                \\
        $\precvar$                                         & 1 when theatre $\roomidx \in \roomset$, on day $\planninghorizonidx \in \planninghorizonset$, is allocated only to subspeciality $\specialityidx \in \specialityset$; 0 otherwise  \\
        $\integervar$                                      & Number of transplants to be performed in block $\blockidx \in \blockset$ by subspeciality $\specialityidx \in \specialitytxset$ \\
        $\anestslackvar$                                   & Slack variable for the number of anaesthetists available in block $\blockidx \in \blockset$ \\
        $\teamslackvar$                                    & Slack variable for the number of teams available in block $\blockidx \in \blockset$ for subspeciality $\specialityidx \in \specialityset$ \\
        \bottomrule
    \end{tabular}
\end{table}

The Objective Functions (OFs) are organised according to the priority levels established by the guidelines from the case study --- although they can be reordered as needed to reflect different institutional priorities in other settings.
To reflect these priorities, we adopt a hierarchical optimisation strategy~\citep[e.g.,][]{Aringhieri2022Hierarchical}.
The OFs are optimised sequentially, while respecting the optimal values obtained in the previous stages.

To identify the minimum human resource-availability violation required to obtain a feasible allocation, the first OF minimises the total violation of team and anaesthetist availability. 
Since anaesthetists are a scarcer and more expensive resource in our case study, anaesthetist slack is assigned a higher penalty than team slack. 
For this purpose, we denote $\anestconstant$ as the penalty coefficient associated with anaesthetist slack, defined as $\anestconstant = |\blockset|\cdot|\specialityset|\cdot|\roomset|$. 
This value is chosen to be sufficiently large to ensure that, whenever a solution with zero anaesthetist slack exists, any solution with positive anaesthetist slack is dominated. 
The first OF is therefore defined as
\begin{align}
    \text{Min } \displaystyle F_0 = \sum_{\blockidx \in \blockset} \sum_{\specialityidx \in \specialityset} \teamslackvar + \anestconstant \cdot \anestslackvar. \label{eq:f0}
\end{align}
The optimal value of $F_0$ is then fixed in the subsequent stages of the hierarchical procedure.

In the second OF, we seek to minimise the number of allocations outside the aspirational theatre allocation for each subspeciality, thus promoting the use of preferred theatres.
We denote $\allowedrooms_{\roomidx\blockidx\specialityidx} \in \{0,1\}$ as a parameter that is equal to 1 if theatre $\roomidx \in \roomset$, in block $\blockidx \in \blockset$, is preferred to be allocated to subspeciality $\specialityidx \in \specialityset$ in the block $\blockidx \in \blockset$, and 0 otherwise.
Then, the second OF is defined as
\begin{align}
    \text{Min } \displaystyle F_1 = \sum_{\roomidx \in \roomset} \sum_{\blockidx \in \blockset} \sum_{\specialityidx \in \specialityset} \blockvar  \text{ only for variables where } \allowedrooms_{\roomidx\blockidx\specialityidx} = 0. \label{eq:f1}
\end{align}

The subsequent OF maximises the allocation of priority subspecialities, ensuring that they receive as much block allocation as possible within the constraints of the system.
Therefore, the third OF is defined as
\begin{align}
    \text{Max } \displaystyle F_2 = \sum_{\roomidx \in \roomset} \sum_{\blockidx \in \blockset} \sum_{\specialityidx \in \specialityprset} \blockvar. \label{eq:f2}
\end{align}

In the fourth OF we maximise the total number of allocated blocks, thus promoting overall OT utilisation and service provision.
We define this OF as
\begin{align}
\text{Max } \displaystyle F_3 = \sum_{\roomidx \in \roomset} \sum_{\blockidx \in \blockset} \sum_{\specialityidx \in \specialityset} \blockvar. \label{eq:f3}
\end{align}

Finally, the fifth OF promotes the allocation of sequential blocks to the same subspeciality, which can improve the logistics of the surgical process and reduce the need for equipment and team transfers between theatres.
Additionaly, this objective is also important for the case where surgeries from the subspeciality allocated to the morning block are delayed, as the afternoon block can be used to continue operating on patients from the same subspeciality, potentially improving efficiency compared with switching to a different subspeciality.
This last OF is defined as
\begin{align}
    \text{Max } \displaystyle F_4 = \sum_{\roomidx \in \roomset} \sum_{\planninghorizonidx \in \planninghorizonset} \sum_{\specialityidx \in \specialityset} \precvar. \label{eq:f4}
\end{align}

Based on this OF structure, we define the \modelname~(Long-Term OT Scheduling) model.
It incorporates $F_0$ to $F_4$ and associates them with constraints from a real-world scenario, but with the flexibility to accommodate other settings as needed.

\begin{alignat}{4}
    (\modelname) \quad & \omit\rlap{Optimize $F_0, F_1, F_2, F_3, F_4$}\nonumber\\
    & \mbox{subject to} && \quad & \sum_{\roomidx \in \roomset} \sum_{\blockidx \in \blockset} \blockvar  & \leq \upperblockdemand                    &\qquad & \forall\, \specialityidx \in \specialityset  \label{eq:c1}\\
    &                  && \quad & \sum_{\roomidx \in \roomset} \sum_{\blockidx \in \blockset} \blockvar  & \geq \lowerblockdemand   &\qquad & \forall\, \specialityidx \in \specialityset \label{eq:c2}\\
    &                  && \quad & \sum_{\specialityidx \in \specialityset} \blockvar                       & \leq 1     &\qquad & \forall\, \roomidx \in \roomset,\, \blockidx \in \blockset \label{eq:c3}\\
    &                  && \quad & \sum_{\roomidx \in \roomset} \sum_{\specialityidx \in \specialityset} \needanest_{\specialityidx} \cdot \blockvar   & \leq \numbersanests_{\blockidx} + \anestslackvar &\qquad & \forall\, \blockidx \in \blockset \label{eq:c4}\\
    &                  && \quad & \sum_{\roomidx \in \roomset} \blockvar   & \leq \specteams_{\blockidx,\specialityidx} + \teamslackvar  &\qquad & \forall\, \,\blockidx \in \blockset,\,\specialityidx \in \specialityset \label{eq:c5}\\
    &                  && \quad & \sum_{\roomidx \in \roomset} \blockvar   & = 2 \cdot \integervar  & \qquad & \forall\, \blockidx \in \blockset,\,\specialityidx \in \specialitytxset  \label{eq:c9}\\
    &                  && \quad & x_{\roomidx,\blockidx_1(\planninghorizonidx),\specialityidx} + x_{\roomidx,\blockidx_2(\planninghorizonidx),\specialityidx} - 1   & \leq \precvar   &\qquad & \forall\, \roomidx \in \roomset,\,\planninghorizonidx \in \planninghorizonset,\, \blockidx \in \blockset,\,\specialityidx \in \specialityset \label{eq:c6}\\
    &                  && \quad & x_{\roomidx,\blockidx_1(\planninghorizonidx),\specialityidx} + x_{\roomidx,\blockidx_2(\planninghorizonidx),\specialityidx}       & \geq 2 \cdot \precvar   &\qquad & \forall\, \roomidx \in \roomset,\,\planninghorizonidx \in \planninghorizonset,\, \blockidx \in \blockset,\,\specialityidx \in \specialityset,\label{eq:c7}\\
    &                  && \quad & x_{\roomidx,\blockidx_1(\planninghorizonidx),\specialityidx}       & = x_{\roomidx,\blockidx_2(\planninghorizonidx),\specialityidx}   &\qquad & \forall\, \roomidx \in \roomset,\,\planninghorizonidx \in \planninghorizonset,\,\blockidx \in \blockmset,\,\specialityidx \in \specialitymset\label{eq:c8}\\
    &                  &&       & \blockvar                           & \in \{0,1\}                 &        & \forall\, \roomidx \in \roomset,\, \blockidx \in \blockset,\, \specialityidx \in \specialityset\nonumber \\
    &                  &&       & \precvar                            & \in \{0,1\}                 &        & \forall\, \roomidx \in \roomset,\, \planninghorizonidx \in \planninghorizonset,\, \specialityidx \in \specialityset\nonumber \\
    &                  &&       & \integervar                         & \in \mathbb{Z_{+}}          &        & \forall\, \blockidx \in \blockset,\, \specialityidx \in \specialityset\nonumber\\
    &                  &&       & \anestslackvar                      & \in \mathbb{Z_{+}}          &        & \forall \, \blockidx \in \blockset\nonumber\\
    &                  &&       & \teamslackvar                       & \in \mathbb{Z_{+}}.          &        & \forall \, \blockidx \in \blockset,\, \specialityidx \in \specialityset \nonumber
\end{alignat}

\bigskip

Constraints~\eqref{eq:c1} and~\eqref{eq:c2} impose lower and upper bounds on the number of allocations of each subspeciality.
These constraints link the tactical scheduling model to the strategic layer, as these bounds are preprocessed information (see Section~\ref{subsec:preprocess_otsp}) coming from the $(R,Q)$ policy and the cancellation-risk management. 

Constraint~\eqref{eq:c3} ensures that each theatre and block receives at most one subspeciality.
Constraint~\eqref{eq:c4} is a soft constraint that limits the number of allocations that require anaesthetists ($\needanest_{\blockidx} = 1$) according to their availability in each block ($\numbersanests_{\blockidx}$), while allowing for some flexibility through the slack variable $\anestslackvar$.
Similarly, constraint~\eqref{eq:c5} limits the number of allocations for each subspeciality based on the availability of medical teams defined by $\specteams_{\blockidx\specialityidx}$, with the slack variable $\teamslackvar$ providing flexibility when necessary.

Constraint~\eqref{eq:c9} ensures that transplant subspecialities are allocated in pairs, reflecting the need for simultaneous OT use for donor and recipient procedures.
Constraints~\eqref{eq:c6} and~\eqref{eq:c7} define $\precvar$ as an indicator of whether theatre $\roomidx$ is allocated to subspeciality $\specialityidx$ in both sequential blocks $\blockidx_1(\planninghorizonidx)$ and $\blockidx_2(\planninghorizonidx)$ of day $\planninghorizonidx$.
This allocation pattern is then incentivised by OF $F_4$.
Constraint~\eqref{eq:c8} enforces this same sequential allocation for subspecialities that necessarily require two consecutive blocks.

No explicit constraints are required to model theatre-subspeciality compatibility. 
Allocation variables $\blockvar$ are introduced only for compatible pairs $(\roomidx,\specialityidx)$, where theatre $\roomidx \in \roomset$ provides the infrastructure required by subspeciality $\specialityidx \in \specialityset$. 
Hence, infeasible assignments are excluded a priori from the feasible region, and all feasible solutions automatically satisfy compatibility requirements.
\section{Computational Study\label{sec:experiments}}
We conduct a computational study to evaluate the performance of the proposed approach using real-world data from \hcunicamp. 
While the theoretical results in this paper provide guarantees on the framework, it is essential to assess the practical implications in a realistic setting.
To this end, we simulate the long-term surgery planning process and compare the proposed approach with a baseline policy in terms of patient waiting-lists and resource utilisation.

\subsection{Real-world Data\label{instances}}
The dataset comprises detailed information from HC-Unicamp on surgery durations, OT availability, medical team constraints, and subspeciality schedules over a one-year period: from January to December 2025.

\subsection{Derivation of Model Inputs}
The data was processed to derive demand distributions, surgery duration distributions, and tactical parameters.

\paragraph{\textbf{Demand distributions}}
Patient arrivals are modelled as Poisson processes. 
For each subspeciality, the arrival rate $\lambda$ is estimated as the average monthly demand observed in the historical activity data. 
The resulting demand distributions are then used as inputs to the \rqmodule~module to determine the $(R,Q)$ policy parameters. 
They are also used to build an initial queue for each subspeciality by sampling six monthly demand values, thereby simulating six months of patient arrivals before the first scheduling decision is made.

\paragraph{\textbf{Surgery duration distributions}}
For each subspeciality, empirical distributions of surgery durations are estimated from historical data. 
These distributions are used in the \overtimemodule~module to evaluate the newsvendor model and determine the number of surgeries to be scheduled per block; see Section \ref{sec:module2}.

\paragraph{\textbf{Tactical parameters}}
All parameters of the \modelname~model (see Table~\ref{tab:symb} in Section~\ref{sec:ot}) are based on real-world data, i.e., the actual values observed in 2025. These include theatre compatibilities, staff availability, and working calendars.

\subsection{Experimental Setup\label{sec:experimental_setup}}
Experiments were conducted in an Ubuntu 24.04 machine equipped with an Intel Core 7 150U processor (5.40 GHz) and 32 GB of RAM. 
All modules are implemented in Python 3.12, and the MILP model is solved using SCIP 9.2.3 with default parameter settings, interfaced through the Pyomo library.
The source code is publicly available at \url{https://github.com/jpfsilvaa/long_term_management_}. The real-world data used in this study are subject to non-disclosure agreements and therefore cannot be shared publicly or included in the code repository.

For the \rqmodule~module, the waiting time target is set to $\waittimereq = 6$ months with confidence level $\conflevelrq = 0.1$, and the threshold $R$ is defined as the midpoint between $\lowq$ and $\upq$.
In other words, we design the $(R,Q)$ policy to ensure that, with 90\% confidence (i.e., $1 - \conflevelrq$), patients are treated within 6 months from the time they join the queue, while the threshold $R$ is set to the average of the lower and upper bounds of the control interval, providing a balanced point for triggering scheduling decisions.

The parameters for the \overtimemodule~module are set as follows. 
The unit idle-time cost, $\excesscost$, is set to $1$, whereas the overtime penalty cost, $p$, is set to $2 \cdot \excesscost$.
Under this specification, and analogously to the example at the end of Section~\ref{sec:cancel}, the optimal time budget is given by $\safetythreshold = F^{-1}_{\rvsurgdur} \left( \frac{2}{3} \right)$, which implies that at least two thirds of the green-lit surgeries are expected to be completed without incurring overtime.
In the inventory model step of this module, the marginal gain of scheduling an additional surgery in a block, denoted by $g$, is set to $2$.
The cancellation cost function is specified as $\cancellationfunction(x) = \monetarycancconst \cdot x^2$, where $x$ denotes the number of surgeries cancelled in a block, and $\monetarycancconst$ is a monetary constant, set to $1$ in this study.
This quadratic cost function reflects the increasing marginal cost of cancellations, capturing the operational and patient-related consequences of cancelling multiple surgeries in the same block.


\subsection{Simulation Framework}

Figure \ref{fig:longterm_simulation} illustrates the simulation flow for the numerical example. The system evolves over a monthly decision cycle. In each period, patient arrivals are realised according to the estimated demand distributions, and queues are updated accordingly. Only subspecialities with queue size exceeding $R$ are selected for scheduling.

The MILP model then determines the allocation of blocks. The number of surgeries performed in each scheduled block is sampled from the probability distribution obtained from the \overtimemodule~module - Eq. \eqref{eq:distnew}. 
%
Queues are updated after the monthly cycle, and the process is repeated at each new period.

\begin{figure}[htb]
    \centering
    \includegraphics[width=0.99\textwidth]{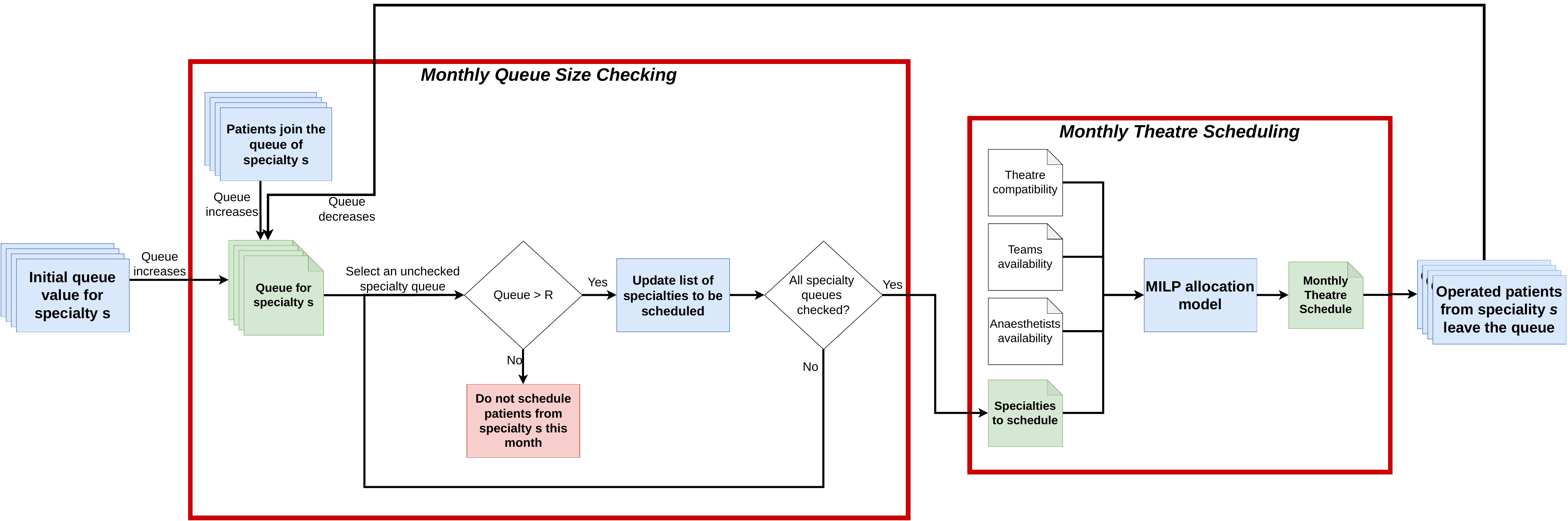}
    \caption{Framework simulation flow. The proposed approach follows a monthly decision cycle in which scheduling decisions are guided by strategic modules designed to control the long-term evolution of waiting lists and the risk of cancellations.}
    \label{fig:longterm_simulation}
\end{figure}

The baseline simulation flow, illustrated in Figure~\ref{fig:baseline_experiment}, follows a similar structure, but without the strategic modules. 
In this case, the policy allocates the subspecialities with no strategic waiting time oversight, using the same MILP model but without the constraints derived from the strategic modules (Eq.~\eqref{eq:c1} and Eq.~\eqref{eq:c2}).
For each scheduled block, the number of patients operated on is sampled from the distribution obtained by recursively convolving the original surgery duration distributions up to $n$ times, where $n$ is the maximum number of surgeries that can be completed within a single block at the chosen confidence level. 
Overtime and surgery cancellations are not considered when constructing this distribution.
This allows us to isolate the impact of the strategic modules on long-term system performance, while keeping hospital parameters and patient arrival process consistent across both approaches.
\begin{figure}[htb]
    \centering
    \includegraphics[width=0.8\textwidth]{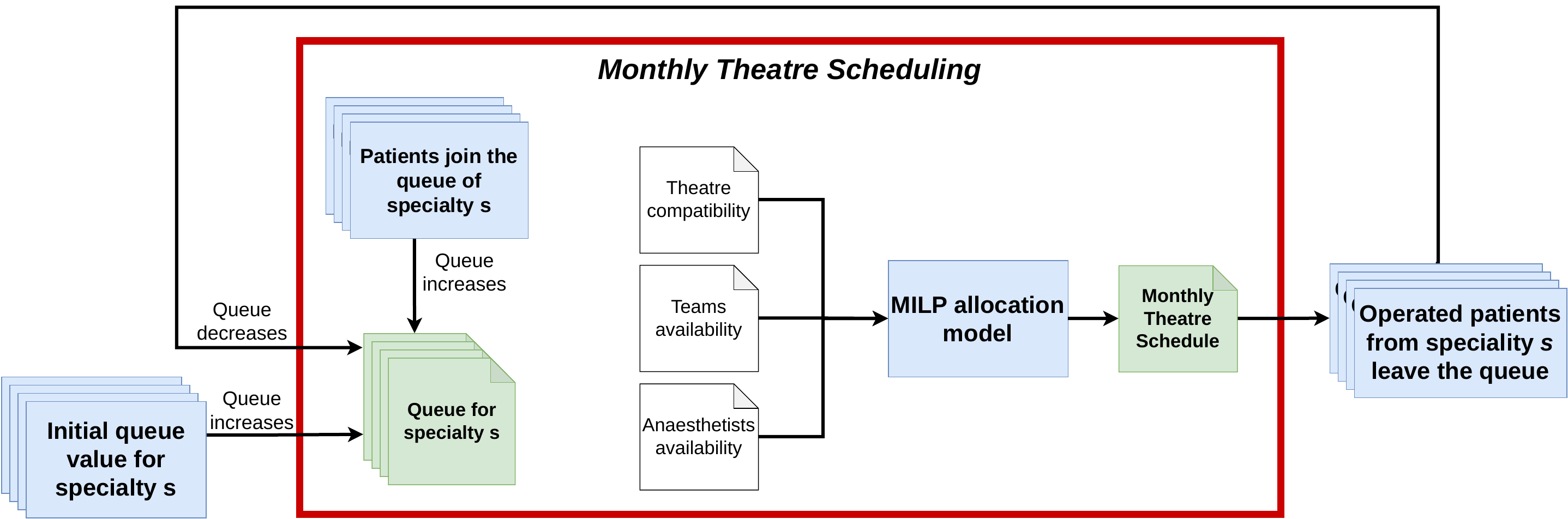}
    \caption{Baseline policy simulation flow. The baseline policy follows a simulation path similar to that of the proposed approach but excludes the strategic modules; queues enter the simulation only through patient arrivals and dispatching decisions, rather than informing the planning decisions.}
    \label{fig:baseline_experiment}
\end{figure}

To assess both policies, we perform $100$ independent simulation runs each, using different random seeds to capture the stochasticity in surgery durations and patient arrivals.
This way, we can evaluate not only the expected performance, but also the variability and robustness of the proposed approach compared to the baseline policy.

\subsection{Results from Strategic Decision-Making\label{sec:longterm_analysis}}
This section reports and analyses the results obtained from the strategic components of the proposed framework, namely the \rqmodule~and the \overtimemodule~modules. 
Although these two components are solved independently, their outputs are combined to support tactical decisions in the OT Scheduling module. 
This decomposition has an important computational implication: the strategic modules are computed once, and their outputs are then used as fixed inputs by the scheduling module, which is executed at each monthly cycle of the simulation horizon.

We first discuss the results of the \overtimemodule~module, which estimates the number of surgeries that can be planned in each block while controlling the risk of cancellation. 
We then present the results of the $(R,Q)$ policy, which defines the control intervals for each subspeciality. 
These intervals are reported after conversion into block units (see Section~\ref{subsec:preprocess_otsp}), using the probability distribution of the optimal number of surgeries planned per block obtained from the \overtimemodule~module. 
This conversion allows the control intervals to be compared across subspecialities in a common unit that is directly relevant to the scheduling module. 
In practice, the number of patients that can be treated in a block depends on both the subspeciality-specific distribution of surgical durations and the cancellation rules defined by the \overtimemodule~module.

\subsubsection{Overtime and Cancellation Results\label{sec:overtime_results}}
This section reports the results of the \overtimemodule~module for Foot Orthopaedics, which is used as a representative subspeciality, although we perform the same analysis for all subspecialities. 
Figure~\ref{fig:newsvendor_pmf} shows the probability distribution of the number of surgeries completed in a Foot Orthopaedics block for different planned values of $\plannedsurgs$. 
We obtain the distributions after applying the newsvendor-based cancellation rule defined in Section~\ref{sec:cancel}, according to which an additional surgery is performed only if the remaining session time is sufficiently large.

\begin{figure}[htb]
    \centering
    \includegraphics[width=0.6\textwidth]{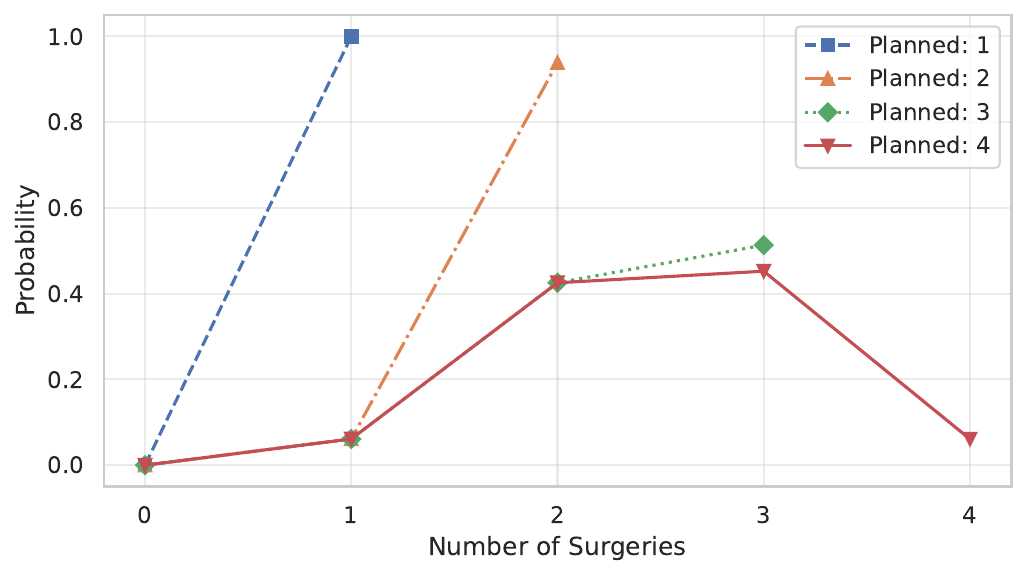}
    \caption{Probability distribution of completed surgeries in Foot Orthopaedics blocks for different planned values of $\plannedsurgs$.}
    \label{fig:newsvendor_pmf}
\end{figure}

The results reveal a clear trade-off between block productivity and cancellation risk. 
When only one surgery is planned in a session ($\plannedsurgs=1$), the surgery is always completed, and the cancellation probability is zero. 
Although this planning decision is operationally safe, it underuses the available block capacity for this subspeciality. 
When two surgeries are planned ($\plannedsurgs=2$), the probability mass remains highly concentrated around two completed surgeries, indicating that the block can usually accommodate this workload without substantial cancellation risk. 
When three surgeries are planned ($\plannedsurgs=3$), the distribution becomes more dispersed: completing all three surgeries remains feasible, but the probability of completing only two surgeries becomes comparatively relevant. 
This indicates that the third planned surgery increases expected throughput while introducing a non-negligible cancellation risk. 
Finally, when four surgeries are planned ($\plannedsurgs=4$), the probability of completing all planned surgeries becomes very low, and the realised output is more likely to remain below the nominal plan.

These findings highlight that a theatre block should not be treated as a deterministic capacity unit capable of accommodating a fixed number of surgeries with certainty. 
Instead, increasing the number of planned surgeries changes the probability distribution of the realised surgical output. 
Consequently, using only nominal capacity would overestimate the effective contribution of a block, particularly when many surgeries are assigned to the same session. 
This insight is central to the proposed framework: while the $(R,Q)$ policy may indicate that a certain number of patients should be released for treatment during a monthly cycle, the Overtime and Cancellation module determines how much of this demand can realistically be assigned to each block without creating excessive cancellation risk.

Figure~\ref{fig:gain_cost} complements this analysis by comparing the total gain and the expected cancellation cost associated with each value of $\plannedsurgs$ for the same subspeciality.
While Figure~\ref{fig:newsvendor_pmf} shows the probability distributions of the possible numbers of surgeries to plan, Figure~\ref{fig:gain_cost} translates these values into a decision criterion for selecting the optimal number of surgeries per block, as defined by the inventory-based model in Section~\ref{sec:patientsperblock}.

\begin{figure}[htb]
    \centering
    \includegraphics[width=0.6\textwidth]{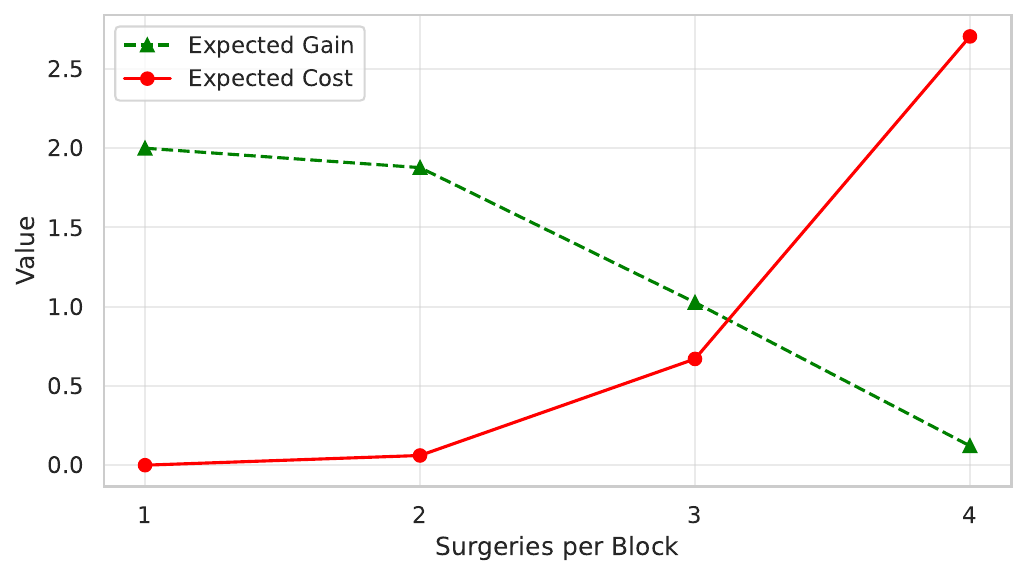}
    \caption{Expected total gain and expected cancellation cost associated with planning $\plannedsurgs$ surgeries in a Foot Orthopaedics block.}
    \label{fig:gain_cost}
\end{figure}

The gain--cost comparison confirms that planning too few surgeries underuses available OT capacity, whereas planning too many surgeries produces an unfavourable cancellation profile. 
For $\plannedsurgs=1$ and $\plannedsurgs=2$, the gain is substantially higher than the expected cost, indicating that these planning levels are operationally conservative. 
However, they also leave potential throughput unused, particularly when compared with the additional output that can be obtained by planning a third surgery. 
For $\plannedsurgs=3$, the expected cancellation cost increases, but remains below the gain. 
This suggests that the third planned surgery is still justified, as it improves expected block productivity without making cancellation risk dominant.

The main change occurs when moving from $\plannedsurgs=3$ to $\plannedsurgs=4$. 
At this point, the gain drops sharply, while the expected cancellation cost increases substantially. 
This pattern indicates that the fourth planned surgery contributes little to the expected number of completed surgeries and mainly increases the likelihood of cancellations. 
Therefore, for Foot Orthopaedics, the experimental results support $\optplannedsurgs=3$ as the recommended number of surgeries to plan per block. 

The \overtimemodule~results show that the effective capacity of each block is subspeciality-specific and depends on the empirical distribution of surgical durations. 
For Foot Orthopaedics, three planned surgeries per block provide the best balance between productivity and cancellation risk. 
This output is then combined with the $(R,Q)$ limits to determine how many blocks should be allocated to this subspeciality in each monthly cycle to ensure adherence with the long-term waiting time requirements - Section \ref{subsec:preprocess_otsp}. 
Therefore, the MILP does not allocate blocks on the basis of nominal session availability. 
Instead, it uses a well-structured probabilistic estimate of how many surgeries each allocated block can realistically deliver while controlling overtime-related cancellations.

\subsubsection{$(R,Q)$ Batch Size Results}

This section reports the control intervals obtained from the $(R,Q)$ policy for all subspecialities of the case study.
Figure~\ref{fig:rq_intervals} summarises these intervals after conversion into block units (see Section~\ref{subsec:preprocess_otsp}). 
To improve readability, subspecialities are ordered according to interval width, $\upperblockdemand - \lowerblockdemand$, rather than alphabetically.

\begin{figure}[htb]
    \centering
    \includegraphics[width=0.9\textwidth]{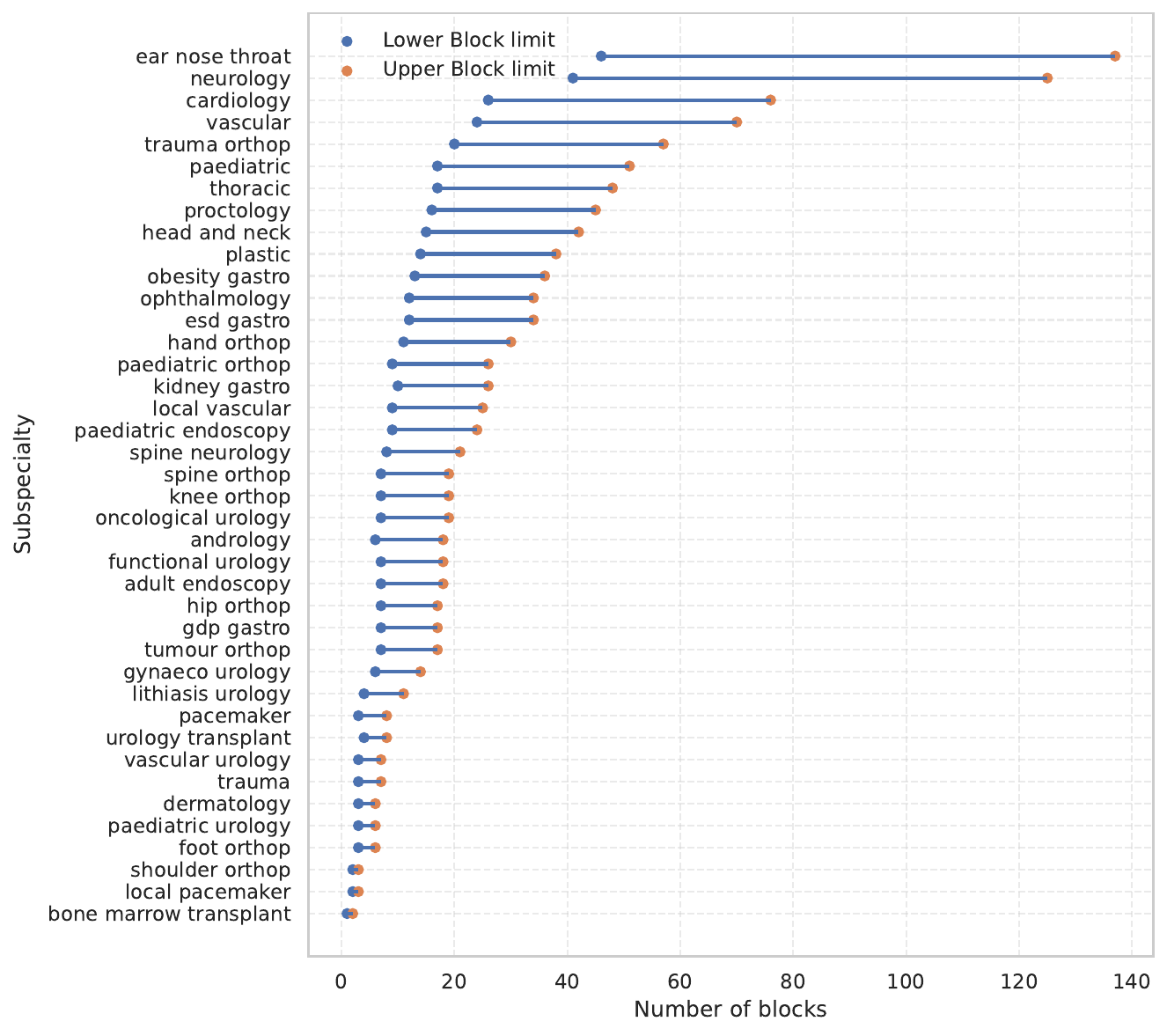}
    \caption{Control intervals for the $(R,Q)$ policy across subspecialities. The intervals represent the admissible range of blocks that can be allocated to each subspeciality when its queue exceeds the reorder point $R$.}
    \label{fig:rq_intervals}
\end{figure}

The results reveal substantial heterogeneity across subspecialities. 
Although most services are associated with relatively narrow feasible intervals, a smaller group exhibits considerably higher lower and upper limits. 
This pattern is consistent with variation in patient demand distributions. 
Subspecialities with higher expected demand and high surgical time variability tend to produce higher values of both $\lowerblockdemand$ and $\upperblockdemand$, whereas lower variability tends to generate shorter intervals.

High-interval subspecialities, such as ENT (Ear, Nose and Throat), NEURO (Neurosurgery), and VASC (Vascular Surgery), provide greater flexibility because the number of blocks, and consequently the number of patients released for scheduling, can vary substantially from one planning cycle to another without compromising waiting-time requirements. 
From a decision-making perspective, this flexibility is particularly important for broad-access or high-referral subspecialities, which may accumulate patients rapidly when surgical capacity is constrained. 
At the same time, these subspecialities may release significant block capacity to other services when their own queues are under control.

By contrast, low-interval subspecialities operate within tighter allocation ranges. 
This does not imply lower clinical importance. 
Rather, it indicates that their demand is smaller and more stable from the perspective of queue replenishment. 
Subspecialities in the medium group occupy an intermediate position: their demand is recurrent enough to require some flexibility, but not sufficiently large or variable to dominate the release process.

These results have direct implications for capacity management planning. 
Because the $(R,Q)$ limits determine the number of patients transferred from the queue-management layer to the tactical planning stage, subspecialities with wider and higher intervals are more likely to influence block allocation. 
This does not necessarily mean that these services should always receive more operating-theatre time. 
Instead, it indicates that their larger and more variable demand must be explicitly represented in the tactical optimisation problem. 
Conversely, subspecialities with narrow intervals generate smaller and more predictable demand releases, reducing their impact on capacity allocation decisions.

It is also important to recall that the $(R,Q)$ policy triggers allocation decisions for a subspeciality only when its queue size exceeds the threshold $R$. 
Thus, even subspecialities with wider intervals will not necessarily dominate the block allocation process, because they are considered for scheduling only when their queues are sufficiently large. 
This feedback mechanism ensures that the system remains responsive to actual demand conditions, rather than being driven solely by potential variability in patient arrivals.

In summary, the control intervals provide a quantitative link between queue dynamics and tactical planning. 
They translate subspeciality-specific demand characteristics into structured release limits, ensuring that the scheduling module receives adaptive demand inputs that are consistent with both the statistical behaviour of the queues and the institutional characteristics of a high-demand public hospital such as \hcunicamp.

\subsection{Long-term Stability Analysis\label{sec:otscheduling}}

This section evaluates the long-term stability of the system under the proposed framework and compares it with the baseline policy. 
The analysis focuses on two complementary dimensions: the evolution of patient queues over time and the monthly allocation of operating-theatre capacity. 

\subsubsection{Queue Management Results}

The one-year simulation horizon provides a detailed view of queue dynamics under both policies. 
Results are reported using the 10--90\% percentile interval to characterise variability across simulation replications.

Figure~\ref{fig:queue_evolution_eda} presents the queue evolution for \textit{Paediatric Endoscopy}. 
Under the proposed framework (Figure~\ref{fig:eda_proposed}), the system exhibits two distinct phases. 
The first is a transient phase, during which the queue decreases steadily from its initial level. 
This is followed by a controlled regime, in which the queue starts to stabilise, approximately between months 6 and 9, and then fluctuates around the reorder threshold $R$. 
This transition reflects the action of the $(R,Q)$ policy, which introduces a demand-driven feedback mechanism: allocation decisions are triggered only when the queue exceeds $R$, thereby preventing unnecessary service once the waiting list has been reduced to a desired point. 
As a result, the system converges towards a stable operating region that balances waiting-time requirements with the hospital capacity constraints.

\begin{figure}[htb]
    \centering
    \begin{subfigure}[b]{0.49\textwidth}
        \includegraphics[width=\textwidth]{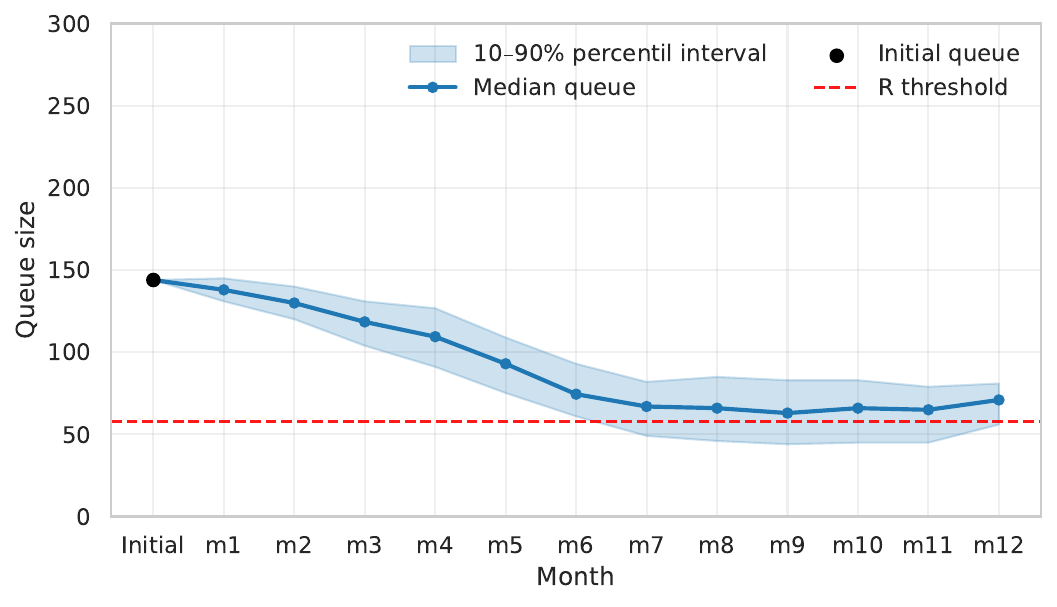}
        \caption{Queue evolution under the proposed framework.}
        \label{fig:eda_proposed}
    \end{subfigure}
    \hfill
    \begin{subfigure}[b]{0.49\textwidth}
        \includegraphics[width=\textwidth]{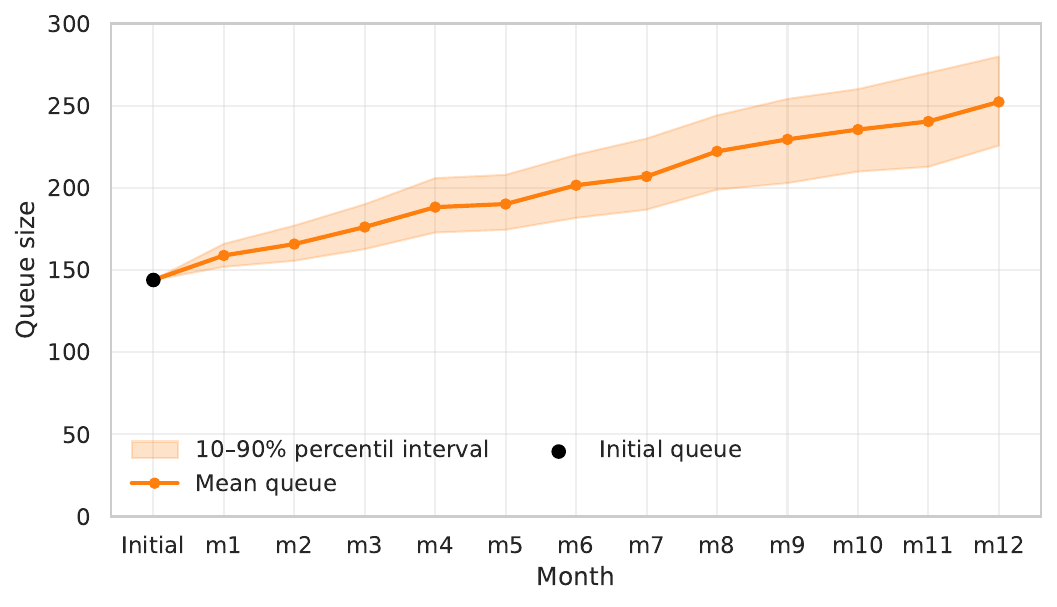}
        \caption{Queue evolution under the baseline policy.}
        \label{fig:eda_baseline}
    \end{subfigure}
    \caption{Evolution of patient queues over time for the subspeciality \textit{Paediatric Endoscopy}.}
    \label{fig:queue_evolution_eda}
\end{figure}

In contrast, the baseline policy leads to a markedly different trajectory (Figure~\ref{fig:eda_baseline}). 
The queue grows throughout the simulation horizon and is accompanied by increasing dispersion across replications. 
This pattern suggests a persistent mismatch between demand and allocated capacity for this subspeciality. 
As the baseline allocation rule does not explicitly account for the current queue state, it may fail to prioritise subspecialities with growing queues, leading to a feedback loop in which demand continues to accumulate without sufficient service.

\begin{figure}[htb]
    \centering
    \begin{subfigure}[b]{0.49\textwidth}
        \includegraphics[width=\textwidth]{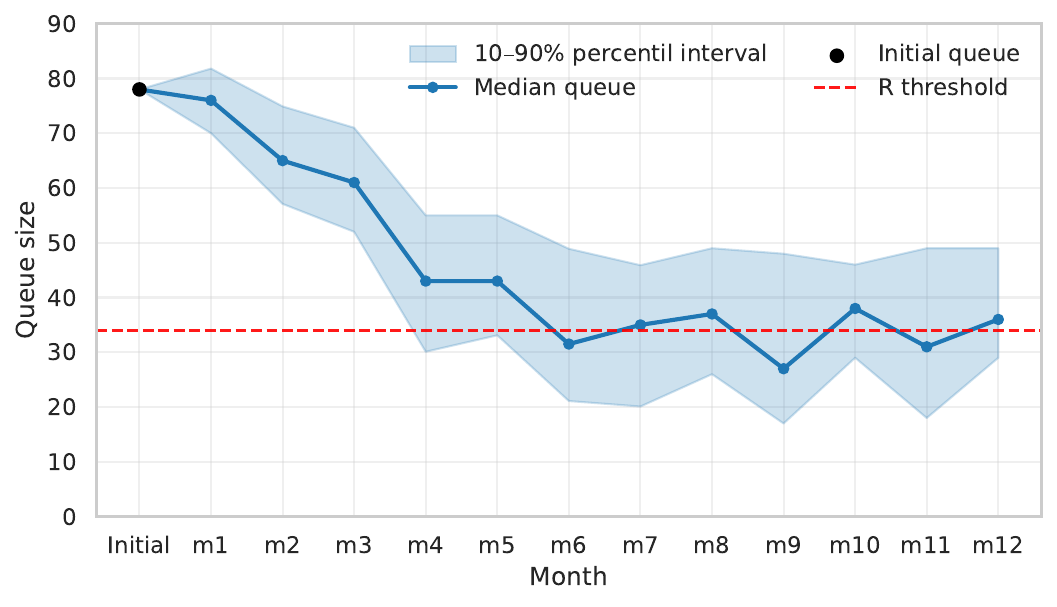}
        \caption{Queue evolution under the proposed framework.}
        \label{fig:gast_proposed}
    \end{subfigure}
    \hfill
    \begin{subfigure}[b]{0.49\textwidth}
        \includegraphics[width=\textwidth]{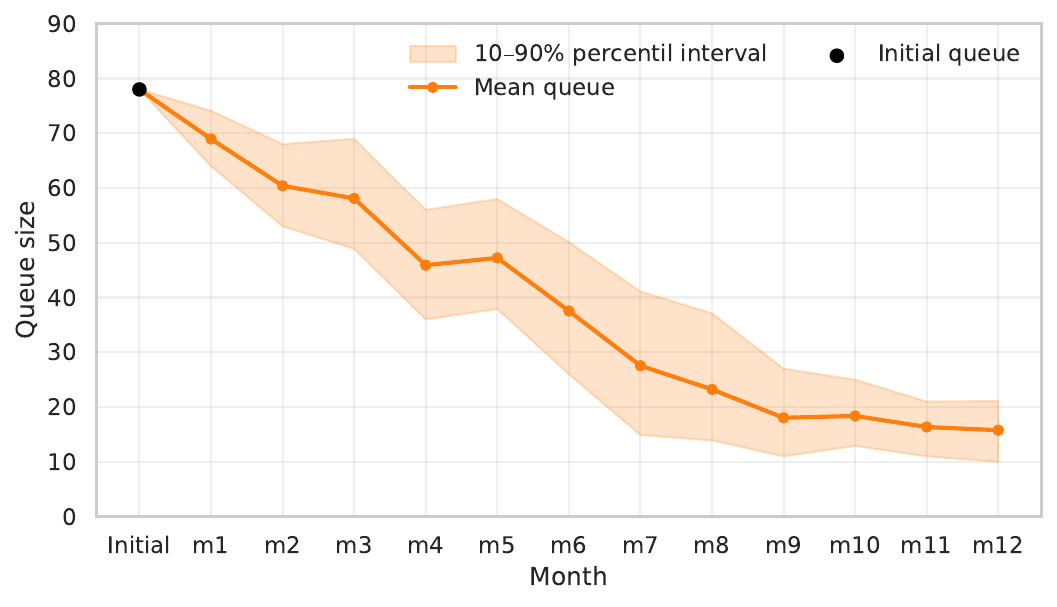}
        \caption{Queue evolution under the baseline policy.}
        \label{fig:gast_baseline}
    \end{subfigure}
    \caption{Evolution of patient queues over time for the subspeciality \textit{Gastroenterology -- Esophagus-Stomach-Duodenum (ESD)}.}
    \label{fig:queue_evolution_gast_eed}
\end{figure}

A similar stability pattern is observed for \textit{Gastroenterology -- Esophagus-Stomach-Duodenum (ESD)} in Figure~\ref{fig:queue_evolution_gast_eed}. 
Under the proposed framework (Figure~\ref{fig:gast_proposed}), the initial waiting list is reduced rapidly, and the queue subsequently remains close to the reorder threshold. 
This indicates that the $(R,Q)$ mechanism adapts allocation decisions to the evolving state of each subspeciality while maintaining a consistent target operating region. 
Importantly, the queue does not continue decreasing indefinitely. 
Instead, it stabilises around $R$, which is the intended behaviour of a reorder-point policy designed to regulate, rather than eliminate, queues.

For the baseline policy (Figure~\ref{fig:gast_baseline}), the queue also decreases over time. 
However, this local improvement should not be interpreted as evidence of system-wide stability. 
As the baseline does not prioritise subspecialities according to queue status, it may allocate capacity to services whose queues are already manageable, while other subspecialities remain congested. 
Thus, the baseline can produce favourable trajectories for selected queues while still generating inefficient allocation patterns at the system level. 
This distinction is important: analysing a single subspeciality in isolation may obscure whether a policy is effectively coordinating capacity across the full portfolio of surgical queues.

Figure~\ref{fig:queue_convergence_r} complements the subspeciality-level analysis by reporting the median queue trajectories for all subspecialities, normalised by their corresponding reorder threshold $R$, i.e., $\frac{\mathit{Queue}(m)}{R} \cdot 100$, where $\mathit{Queue}(m)$ is the queue length in monthly cycle $m$. 
Although the baseline policy does not use the $(R,Q)$ policy, normalising queues by $R$ provides a common reference point for comparing the two approaches across subspecialities.
Also, this normalisation enables comparisons across subspecialities with different demand profiles and queue scales.

Under the proposed framework (Figure~\ref{fig:proposed_convergence}), queues progressively concentrate around the 100\% reference line, indicating convergence towards the intended controlled region.
This system-wide pattern shows that the proposed framework regulates queues consistently across all subspecialities. 
Specifically, queues neither grow indefinitely nor remain persistently above their reorder thresholds. 
This indicates that an arriving patient has a high probability of being treated within the target waiting time, as the policy prevents excessive patient accumulation.

\begin{figure}[htb]
    \centering
    \begin{subfigure}[b]{0.49\textwidth}
        \includegraphics[width=\textwidth]{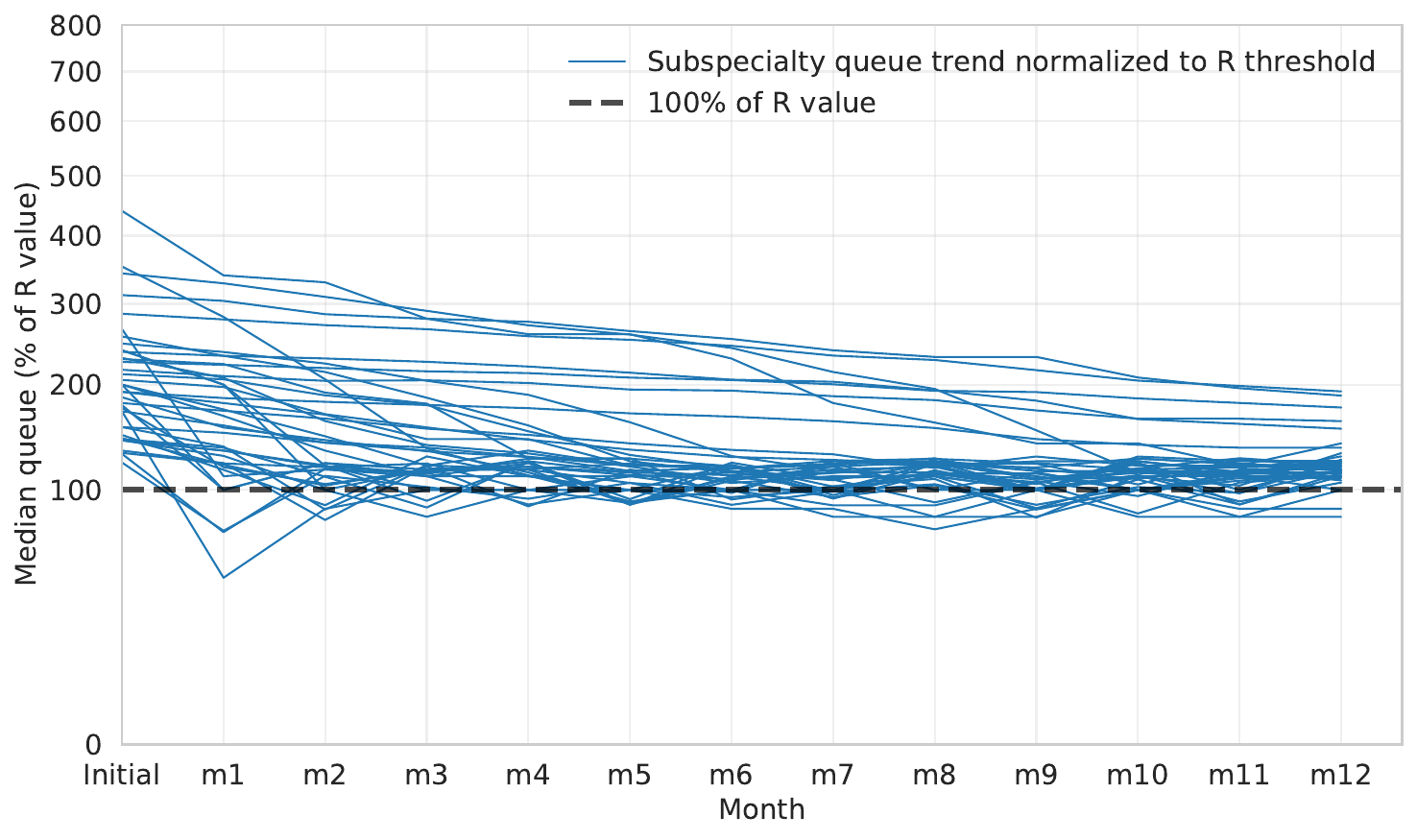}
        \caption{Proposed framework.}
        \label{fig:proposed_convergence}
    \end{subfigure}
    \hfill
    \begin{subfigure}[b]{0.49\textwidth}
        \includegraphics[width=\textwidth]{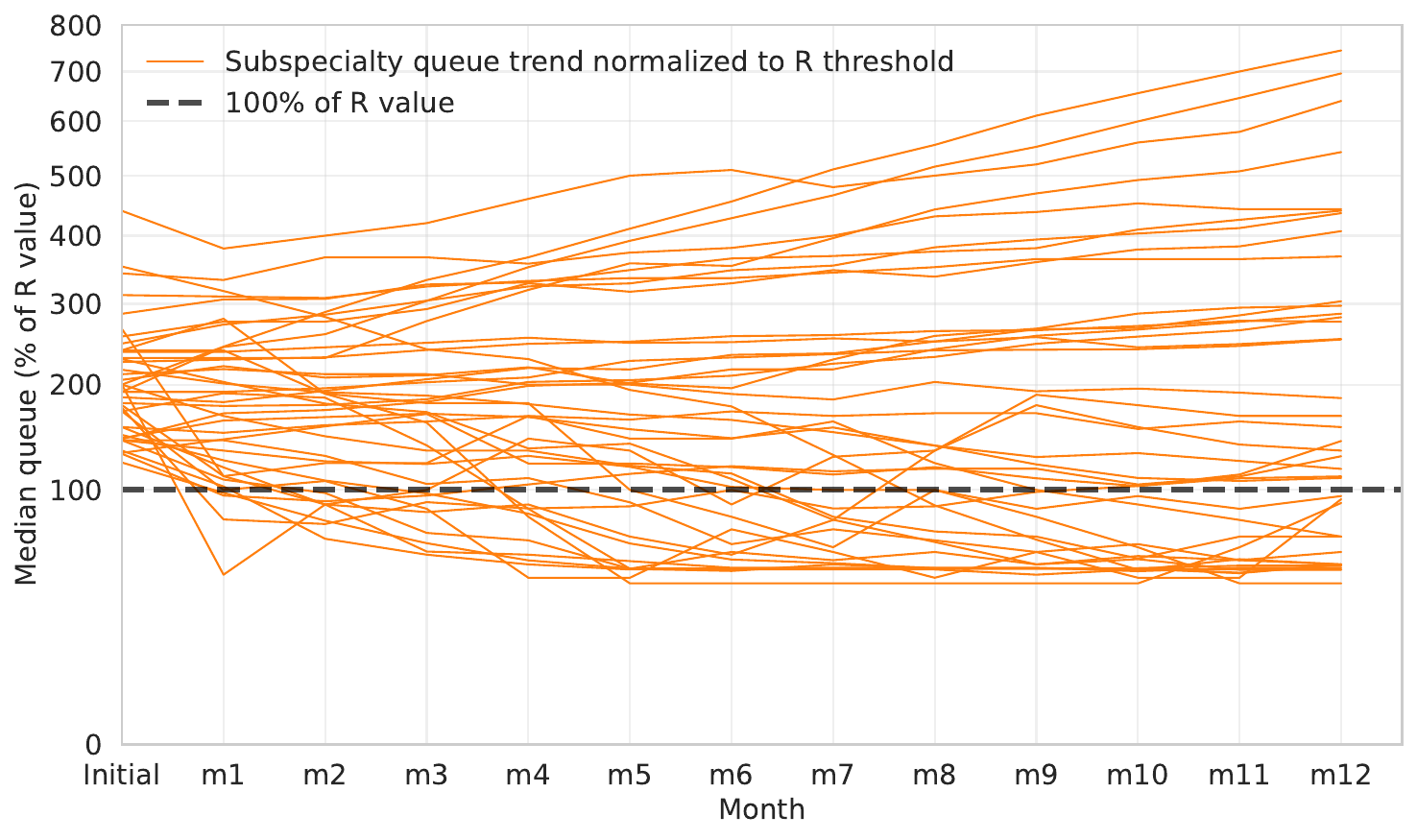}
        \caption{Baseline policy.}
        \label{fig:baseline_convergence}
    \end{subfigure}
    \caption{Median queue trajectories for all subspecialities over the simulation horizon, normalised by the corresponding reorder threshold $R$. The dashed line represents 100\% of the $R$ threshold.}
    \label{fig:queue_convergence_r}
\end{figure}

In contrast, the baseline policy (Figure~\ref{fig:baseline_convergence}) exhibits a more heterogeneous and unstable pattern. 
Some subspecialities move towards lower queue levels, whereas several others experience sustained growth and remain substantially above their reorder thresholds. 
This dispersion reveals the absence of a stabilising mechanism. 
The baseline may perform adequately for selected subspecialities, but it does not ensure consistent performance across the full portfolio of surgical queues. 
More importantly, it may fail to meet waiting-time requirements for many subspecialities, because queues are not explicitly controlled.

In practice, persistent deviations of the kind of the baseline policy would require managerial intervention, either by reducing the intake of patients in the most critical subspecialities, by redefining priority levels to reallocate capacity towards them, or by combining both actions.
Thus, rather than providing a self-regulating planning mechanism, this policy may transfer the burden of restoring queue stability to reactive ad hoc adjustments.

By driving queues towards a controlled operating region, limiting excessive waiting-list accumulation, and supporting demand-driven prioritisation, the proposed framework provides meaningful long-term stability, which is a critical requirement for high-demand public hospitals. 
In contrast, the baseline policy lacks an explicit feedback mechanism based on queue status. 
This can lead to non-compliance with waiting-time requirements and to outcomes that are locally acceptable but globally inefficient. 
Detailed queue trajectories for all subspecialities are provided in~\ref{app:queue_evolution}.

\subsubsection{Long-term OT Allocation Results}

This section analyses the long-term allocation decisions produced by the \otschedmodule~module described in Section~\ref{sec:ot}. 
Figure~\ref{fig:room_occupancy} reports monthly OT occupancy under the proposed framework and the baseline policy. 
Both approaches generate feasible allocations with substantial use of available capacity. 
This is consistent with the MILP formulation, which maximises block allocations while respecting theatre availability, subspeciality compatibility, and monthly capacity limits. 
Specifically, for the proposed framework, this result indicates that the scheduling layer is able to convert planned surgical demand into executable theatre allocations. 
In other words, the MILP successfully implements the strategic decisions produced by the queue-control and overtime modules, translating them into meaningful theatre plans.

\begin{figure}[htb]
    \centering
    \includegraphics[width=0.9\textwidth]{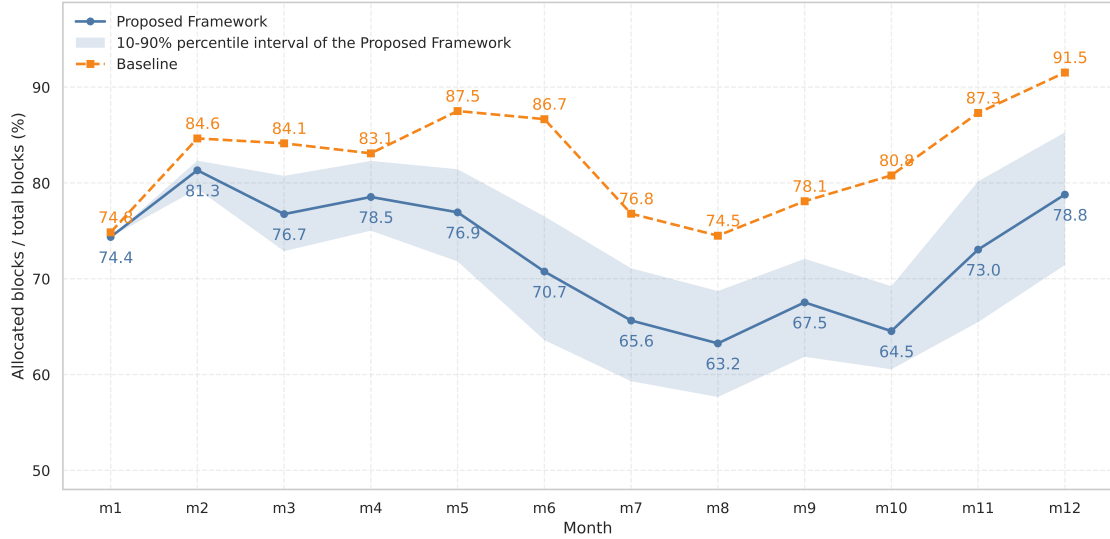}
    \caption{Monthly operating-theatre occupancy under the proposed framework and the baseline policy. Occupancy is calculated as the number of allocated blocks in each month divided by the total number of available blocks. The baseline policy has no shaded area because its allocation decisions are independent of realised demand and queue states. As it attempts to allocate blocks to all subspecialities in every month, the resulting occupancy is deterministic across replications.}
    \label{fig:room_occupancy}
\end{figure}

The baseline policy systematically presents higher occupancy than the proposed framework. 
This difference is expected and reflects the distinct logic behind the two allocation policies. 
The baseline allocates OT capacity without explicitly considering whether each subspeciality queue is above or below any point. 
As a consequence, it attempts to schedule all subspecialities in every month, which leads to higher occupancy but does not prioritise subspecialities with growing queues.

By contrast, the proposed framework triggers allocation only when a subspeciality queue exceeds its reorder threshold $R$. 
When the queue is below this level, the framework does not force an allocation merely to increase theatre utilisation. 
Therefore, the lower occupancy observed in some months should not be interpreted as an inefficiency of the scheduling model. 
Rather, it reflects a deliberate control mechanism that avoids allocating capacity to subspecialities already operating within the desired region.

Figure~\ref{fig:specialities_allocated} supports this interpretation by showing the number of subspecialities allocated each month. 
The baseline assigns blocks to a relatively stable and high number of subspecialities, ranging from approximately $34$ to $39$ throughout the simulation horizon. 
In contrast, the proposed framework allocates fewer subspecialities in most months, with a more pronounced reduction in the second half of the year --- which is consistent with the initial waiting-list sizes of each subspeciality.
This analysis confirms that the difference in occupancy is directly associated with the queue-control mechanism: as more queues enter the controlled region, fewer subspecialities require active allocation.

\begin{figure}[htb]
    \centering
    \includegraphics[width=0.9\textwidth]{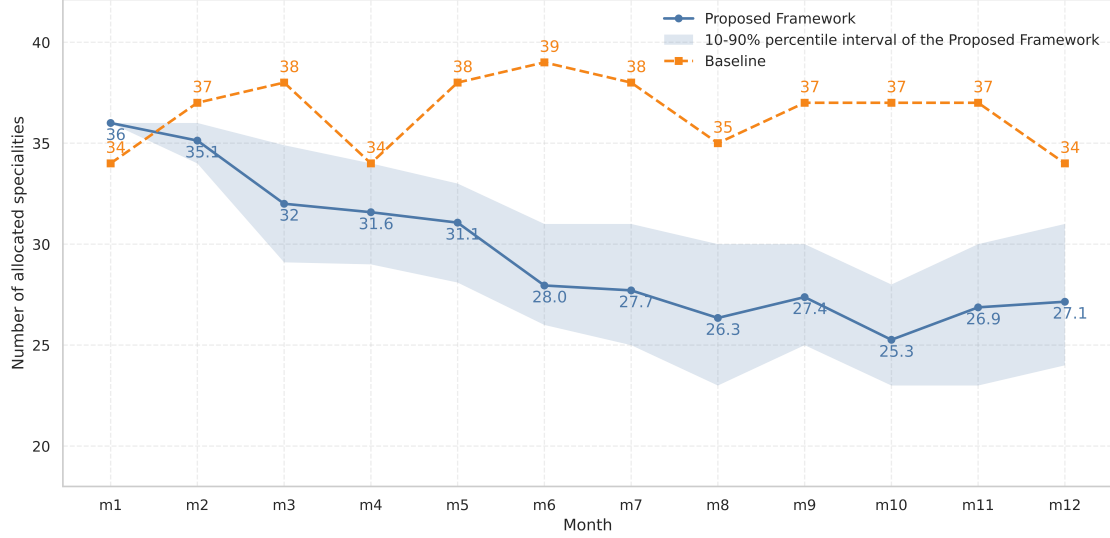}
    \caption{Number of subspecialities allocated per month under the proposed framework and the baseline policy. The baseline policy has no shaded area due to the same reason in Figure~\ref{fig:room_occupancy}: its allocation decisions are independent of realised demand and queue states. As it attempts to allocate blocks to all subspecialities in every month, the resulting occupancy is deterministic across replications.}
    \label{fig:specialities_allocated}
\end{figure}

The behaviour observed in Figure~\ref{fig:specialities_allocated} has important implications for long-term planning. 
A policy that attempts to allocate blocks to subspecialities with no explicit control mechanism may produce higher utilisation rates, but it can also operate patients from queues that are already under control while neglecting those from high accumulation queues. 
The proposed framework introduces a more controlled allocation logic: OT capacity is concentrated on subspecialities whose queues require intervention, while subspecialities below $R$ can temporarily remain unallocated.
This creates a controlled form of flexibility, allowing the OT scheduling MILP to redirect blocks over time according to the evolving state of the system rather than reproducing a fixed allocation pattern. This translates into a more effective use of the available capacity, while also guaranteeing long-term stability.


Additionally, the decision maker can verify that it is not even necessary to acquire more resources to meet waiting-time requirements or to reduce queues.
Instead, the proposed framework can achieve these goals by improving the efficiency of existing resources through better coordination and demand management.
This is a critical insight for high-demand public hospitals, where resource constraints are often a significant barrier to improving access and reducing waiting times. 
By implementing a more strategic and responsive scheduling approach, hospitals can make better use of their existing capacity, potentially reducing the need for costly expansions while still achieving meaningful improvements in patient outcomes and satisfaction.
\section{Conclusions and Future Directions\label{sec:conclusion}}

This paper proposes a modular framework that integrates strategic and tactical decisions for elective surgery planning under uncertainty. 
The framework combines waiting-list control, cancellation risk management, and tactical OT allocation within a comprehensive decision process. 
Rather than solving a single large-scale holistic model, the proposed approach decomposes the resolution into interconnected modules. 
This design is particularly relevant in practice, as it allows decision-makers to focus on specific aspects of the problem while still ensuring that decisions across different levels are aligned with long-term objectives.

We evaluate the proposed approach through a real-world case study in a large Brazilian public referral hospital. 
Results from the \overtimemodule~module provide a quantitative assessment of the trade-off between scheduling additional surgeries and exposing the system to higher cancellation risks. 
The findings indicate that treating surgery durations as stochastic rather than deterministic is crucial for realistic planning. 
Notoriously, when the number of planned surgeries per block increases, the probability of completing all scheduled surgeries decreases, and the risk of cancellations rises significantly.
By using this insight to determine the effective number of surgeries to be planned in a block, the number of surgeries defined by the \rqmodule~module can be translated into realistic block allocation targets that are achievable in practice.

The modified $(R,Q)$ control policy also played a pivotal role in the framework's performance. 
The target batch sizes defined for each subspeciality in the case study provide the necessary flexibility to absorb demand fluctuations.
Additionally, the policy's structure allows for a clear interpretation of the relationship between demand uncertainty, waiting-time requirements, and scheduling targets.
For instance, the results show that subspecialities with higher demand variability may require larger batch sizes to maintain queue stability, while those with more predictable demand can operate with more stable batch intervals. 
Hence, the policy supports a nuanced approach to queue management that is sensitive to the specific characteristics of each subspeciality, rather than applying a generic scheduling rule across the board.

The long-term simulation results further complement the analysis, highlighting the necessity of integrating waiting time management to tactical decisions. 
When OT scheduling ignores waiting-list dynamics and duration uncertainty, the resulting plans may appear attractive regarding short-term utilisation but perform poorly in practice. 
For instance, under a baseline myopic policy, many waiting lists grew uncontrollably in the experiment. 
In contrast, the proposed framework drives all waiting lists towards their respective reorder points $R$ and maintains stability over the simulation horizon. 

Furthermore, the results demonstrate that high theatre occupancy should not be interpreted in isolation: a policy may show high utilisation while allowing queues to deteriorate. Conversely, under the proposed framework, lower average occupancy in certain periods indicates that, once a subspeciality has reached its target queue level, the system can schedule fewer surgeries without jeopardising its waiting-time targets. 
This is a critical insight for hospital administrators, as it suggests that the OT capacity not used by subspecialities below their reorder points can be used for other purposes, such as emergency surgeries or maintenance, or reallocated to other subspecialities that are still above their targets.


This study has limitations that suggest avenues for future research. 
First, demand was modelled as a Poisson process based on historical data. 
This approach does not invalidate the framework, as the strategic modules can be adapted to any alternative demand models, but future applications that have available data about actual patient arrivals could incorporate even more realistic patient arrival processes. 
Second, the framework could be extended to include bed-management decisions in both upstream and downstream units. 
Integrating elective surgeries decisions with ward and ICU availability would help avoid situations in which planned surgeries are assigned to a subspeciality, but patients cannot be operated on because the required beds are unavailable. 
Finally, the framework could incorporate an operational decision layer to address the sequencing of surgeries within blocks, the allocation of specific surgical teams, and the real-time management of emergency cases on the day of surgery.

\section*{Acknowledgments}
\sloppypar
This work was supported by the National Council for Scientific and Technological Development (CNPq) [Grants 312345/2023-2, 404779/2025-5]; and São Paulo Research Foundation (FAPESP) [Grants 2022/05803-3, 2024/02641-8].

\appendix
\section{Evolution of the Queues\label{app:queue_evolution}}
In this appendix, we present the evolution of the queues for each subspeciality under the proposed framework and the baseline policy. 
Table~\ref{tab:queue_evolution} shows the initial and final queue sizes (mean and standard deviation) for each subspeciality, demonstrating how the proposed framework effectively stabilises the queues over time compared to the baseline policy, which leads to significant queue growth in several subspecialities.

\begin{table}[htb!]
    \centering
    \caption{Upper and lower bounds on the batch size $Q$ of the number of patients for each subspeciality, the reorder point $R$, and the initial and final queue sizes (mean and standard deviation) for each subspeciality, under the proposed framework and the baseline policy.\label{tab:queue_evolution}}
    \scriptsize
    \begin{tabular}{lcccccccc}
        \toprule
        \multicolumn{1}{l}{\multirow{2}{*}{\textbf{Subspecialty}}} & \multirow{2}{*}{\textbf{Initial Queue}} & \multirow{2}{*}{\textbf{$\lowq$}} & \multirow{2}{*}{\textbf{$\upq$}} & \multirow{2}{*}{\textbf{R}} & \multicolumn{4}{c}{\textbf{Final Queue}}              \\
        \multicolumn{1}{c}{}                                       &                                         &                                 &                                 &                             & Framework & Framework Std. & Baseline & Baseline Std. \\
        \midrule
        ear nose throat                                   & 330                            & 72                     & 219                    & 146                & 257.05    & 30.90          & 410.59   & 36.06         \\
        trauma orthop                                     & 192                            & 33                     & 100                    & 67                 & 130.54    & 21.01          & 294.73   & 21.73         \\
        neurology                                         & 180                            & 41                     & 125                    & 83                 & 135.33    & 20.52          & 88.57    & 23.85         \\
        paediatric                                        & 150                            & 31                     & 94                     & 63                 & 98.13     & 17.29          & 162.14   & 20.09         \\
        vascular                                          & 144                            & 37                     & 112                    & 75                 & 100.10    & 19.67          & 190.85   & 23.21         \\
        paediatric endoscopy                              & 144                            & 29                     & 87                     & 58                 & 69.70     & 9.61           & 252.41   & 24.44         \\
        cardiology                                        & 108                            & 36                     & 109                    & 73                 & 86.07     & 13.43          & 36.88    & 6.75          \\
        ophthalmology                                     & 108                            & 37                     & 112                    & 75                 & 92.02     & 16.23          & 214.95   & 21.21         \\
        obesity gastro                                    & 108                            & 29                     & 87                     & 58                 & 70.08     & 10.86          & 178.49   & 21.61         \\
        head and neck                                     & 90                             & 25                     & 75                     & 50                 & 62.15     & 8.84           & 141.80   & 18.41         \\
        thoracic                                          & 84                             & 27                     & 81                     & 54                 & 65.69     & 11.47          & 86.64    & 20.74         \\
        knee orthop                                       & 78                             & 13                     & 37                     & 25                 & 47.46     & 13.69          & 92.18    & 11.51         \\
        esd gastro                                        & 78                             & 17                     & 50                     & 34                 & 37.48     & 8.01           & 15.77    & 4.49          \\
        hand orthop                                       & 78                             & 19                     & 56                     & 38                 & 44.21     & 7.81           & 110.97   & 15.32         \\
        paediatric orthop                                 & 72                             & 17                     & 50                     & 34                 & 40.75     & 6.89           & 63.30    & 15.12         \\
        hip orthop                                        & 72                             & 11                     & 31                     & 21                 & 24.44     & 5.82           & 92.52    & 11.05         \\
        lithiasis urology                                 & 60                             & 9                      & 25                     & 17                 & 17.92     & 5.22           & 10.10    & 5.26          \\
        proctology                                        & 60                             & 21                     & 62                     & 42                 & 49.72     & 7.87           & 20.47    & 4.57          \\
        adult endoscopy                                   & 56                             & 16                     & 47                     & 32                 & 40.38     & 7.95           & 46.32    & 14.36         \\
        local vascular                                    & 54                             & 22                     & 65                     & 44                 & 51.51     & 8.09           & 22.82    & 5.04          \\
        gynaeco urology                                   & 54                             & 11                     & 31                     & 21                 & 23.52     & 6.17           & 86.75    & 11.36         \\
        pacemaker                                         & 48                             & 13                     & 37                     & 25                 & 31.74     & 5.54           & 174.23   & 12.93         \\
        plastic                                           & 48                             & 18                     & 53                     & 36                 & 40.21     & 7.45           & 16.84    & 4.54          \\
        andrology                                         & 42                             & 16                     & 47                     & 32                 & 37.00     & 7.81           & 35.83    & 13.59         \\
        spine orthop                                      & 42                             & 15                     & 43                     & 29                 & 36.26     & 6.25           & 34.53    & 12.12         \\
        kidney gastro                                     & 42                             & 14                     & 40                     & 27                 & 30.89     & 6.10           & 13.40    & 3.94          \\
        functional urology                                & 36                             & 11                     & 31                     & 21                 & 24.75     & 5.21           & 11.44    & 4.73          \\
        tumour orthop                                     & 30                             & 11                     & 31                     & 21                 & 24.00     & 3.89           & 13.30    & 6.22          \\
        paediatric urology                                & 30                             & 7                      & 18                     & 13                 & 14.30     & 3.19           & 70.18    & 8.61          \\
        gdp gastro                                        & 30                             & 8                      & 21                     & 15                 & 17.07     & 4.19           & 11.35    & 5.13          \\
        oncological urology                               & 30                             & 11                     & 31                     & 21                 & 25.80     & 5.30           & 21.20    & 9.14          \\
        spine neurology                                   & 30                             & 12                     & 34                     & 23                 & 27.70     & 5.76           & 30.91    & 12.57         \\
        vascular urology                                  & 24                             & 5                      & 12                     & 9                  & 11.13     & 2.75           & 15.11    & 6.91          \\
        dermatology                                       & 22                             & 3                      & 6                      & 5                  & 6.34      & 2.10           & 32.64    & 5.24          \\
        urology transplant                                & 18                             & 5                      & 12                     & 9                  & 10.31     & 3.71           & 66.85    & 6.76          \\
        foot orthop                                       & 18                             & 5                      & 12                     & 9                  & 9.41      & 3.53           & 13.31    & 6.51          \\
        trauma                                            & 16                             & 5                      & 12                     & 9                  & 10.61     & 2.68           & 8.06     & 5.74          \\
        local pacemaker                                   & 12                             & 4                      & 9                      & 7                  & 6.38      & 2.81           & 6.28     & 2.03          \\
        shoulder orthop                                   & 12                             & 3                      & 6                      & 5                  & 5.00      & 1.91           & 2.63     & 1.88          \\
        bone marrow transplant                            & 12                             & 3                      & 6                      & 5                  & 4.49      & 2.24           & 4.63     & 2.38   \\
    \bottomrule
    \end{tabular}
\end{table}

\bibliographystyle{elsarticle-harv}
\scriptsize{\bibliography{ref}}
\end{document}